\documentclass[sigconf]{acmart}
\usepackage{tikz}
\usepackage{booktabs}
\usepackage{graphicx}
\usepackage{xcolor}
\usepackage{algorithm2e}
\usepackage{listings}
\usepackage{tabularx}
\usepackage{flushend}
\usepackage{soul}
\usepackage{wasysym}
\usepackage{multirow}
\usepackage{url}
\usepackage{microtype}
\usepackage{nowidow}
\usepackage{algpseudocode,amsmath}
\usepackage{pseudo}
\let\labelindent\relax
\usepackage{enumitem}
\usepackage{xspace}

\usepackage[acronym]{glossaries}

\setkeys{glslink}{hyper=false}

\newcommand{\ie}{\textit{i}.\textit{e}.,~}
\newcommand{\eg}{\textit{e}.\textit{g}.,~}
\newcommand{\A}{$\mathcal{A}$\xspace}

\newcommand\blfootnote[1]{%
  \begingroup
  \renewcommand\thefootnote{}\footnote{#1}%
  \addtocounter{footnote}{-1}%
  \endgroup
}

\newacronym{ms}{MS}{Malicious Server}
\newacronym{mc}{MC}{Malicious Client}
\newacronym{str}{STR}{Signed Tree Root}
\newacronym{poi}{PoI}{Proof of Inclusion}
\newacronym{pom}{PoM}{Proof of Misbehavior}
\newacronym{akr}{AKR}{Anonymous Key Request}
\newacronym{asr}{ASR}{Anonymous STR Request}

\acmYear{2022}
\acmConference[CCS '22]{Proceedings of the 2022 ACM SIGSAC Conference on Computer and Communications Security}{November 7--11, 2022}{Los Angeles, CA, USA}
\acmPrice{15.00}
\acmDOI{10.1145/3548606.3560588}
\acmISBN{978-1-4503-9450-5/22/11}

\renewcommand\footnotetextcopyrightpermission[1]{} 

 \setcopyright{none}

\begin{document}

\title{Automatic Detection of Fake Key Attacks in Secure Messaging}

\author{Tarun Kumar Yadav}
\affiliation{%
  \institution{Brigham Young University}
   \country{}
  }
  \email{tarun141@byu.edu}

\author{Devashish Gosain}
\affiliation{%
  \institution{Max Planck Institute for Informatics}
    \country{}
  }
  \email{dgosain@mpi-inf.mpg.de}

\author{Amir Herzberg}
\affiliation{%
  \institution{University of Connecticut}
  \country{}}
  \email{amir.herzberg@gmail.com}

\author{Daniel Zappala}
\affiliation{%
  \institution{Brigham Young University}
  \country{}
  }
  \email{zappala@cs.byu.edu }

\author{Kent Seamons}
\affiliation{%
  \institution{Brigham Young University}
  \country{}
  }
  \email{seamons@cs.byu.edu}

\renewcommand{\shortauthors}{Tarun Kumar Yadav, Devashish Gosain, Amir Herzberg, Daniel Zappala, \& Kent Seamons}

\begin{abstract}
Popular instant messaging applications such as WhatsApp and Signal provide end-to-end encryption for billions of users. They rely on a centralized, application-specific server to distribute public keys and relay encrypted messages between the users. Therefore, they prevent passive attacks but are vulnerable to some active attacks. A malicious or hacked server can distribute fake keys to users to perform man-in-the-middle or impersonation attacks. While typical secure messaging applications provide a manual method for users to detect these attacks, this burdens users, and studies show it is ineffective in practice. This paper presents KTACA, a completely automated approach for key verification that is oblivious to users and easy to deploy. We motivate KTACA by designing two approaches to automatic key verification. One approach uses client auditing (KTCA) and the second uses anonymous key monitoring (AKM). Both have relatively inferior security properties, leading to KTACA, which combines these approaches to provide the best of both worlds. We provide a security analysis of each defense, identifying which attacks they can automatically detect. We implement the active attacks to demonstrate they are possible, and we also create a prototype implementation of all the defenses to measure their performance and confirm their feasibility. Finally, we discuss the strengths and weaknesses of each defense, the overhead on clients and service providers, and deployment considerations.
\end{abstract}



\maketitle
\pagestyle{plain}

\section{Introduction}

Secure messaging applications provide billions of users with end-to-end encryption to ensure message privacy. A long list of applications provides this service, including WhatsApp, iMessage, Facebook Messenger, Skype, Signal, Threema, Wire, Wickr, Viber, and Riot. The application's underlying encryption protocols vary, though many use the Signal protocol or some derivation. \blfootnote{An extended version of our paper published at ACM CCS 2022.}

All the secure messaging applications listed above use a centralized server to exchange public keys and relay messages among users. The end-to-end encryption (E2EE) protocols assume the honest-but-curious model. 
When Alice wishes to communicate with Bob, she requests Bob's key from the server (and vice-versa).
A malicious or compromised server can launch a man-in-the-middle (MITM) attack against Alice and Bob by providing them with fake keys. The server then has access to the plaintext as it decrypts and re-encrypts each message that it relays between them.

To help counter these attacks, most of these applications (iMessage and Skype excepted) provide users a method to verify each others' public keys (or derived keys). This verification is typically done by manually comparing a key fingerprint or scanning a QR code of the fingerprint. Most applications do not prompt the users to do this at the start of a conversation but display a prompt if the keys change. Prior studies~\cite{schroder2016signal, herzberg2016can, vaziripour2017you} have found that users are generally oblivious to the need to verify public keys and are unlikely to authenticate, leaving them vulnerable to an attack.

While proof of fake key attacks on secure messaging platforms is hard to obtain, the vulnerability exists and may be exploited. Experience indicates that exploitation is only a matter of time. 
One example of surveillance of a secure messaging app occurred in 2018 when Dutch law enforcement eavesdropped on criminals using the IronChat application~\cite{ironchat}.
Outside of secure messaging, attacks have likewise led to MITM eavesdropping. In 2011, the Iranian government was suspected of obtaining a fraudulent public key certificate to eavesdrop on 300,000 Iranians~\cite{iran} accessing Gmail. Further, the Kazakhstan government recently began using a fake root CA to perform a MITM attack against HTTPS connections to websites including Facebook, Twitter, and Google~\cite{raman2020investigating}. It is also well-known that nations engage in surveillance and would like to
crack secure messaging applications. For example, revelations from leaked documents by Edward Snowden
indicate significant capabilities in the United States regarding surveillance of electronic
communication.

Currently, the only way to detect fake key attacks is to rely on users to perform key verification manually whenever they start a conversation with a contact and any time a contact updates their key.  Legitimate public key updates occur only when contacts re-install the messenger application. Thus, convincing users to always verify keys would almost universally confirm that a key update is legitimate. However, such repeated confirmations might cause user fatigue and a penchant to ignore key update warning messages. Our work aims to relieve users of this burden entirely. 

In this paper, we design and evaluate three novel approaches to detect fake key attacks automatically : (1) Key Transparency with Client Auditors (KTCA, Section~\ref{key_transparency_defense}), (1) Anonymous Key Monitoring (AKM, Section~\ref{Subsec:AKR}) and (3) Key Transparency with Anonymous Client Auditors (KTACA, Section~\ref{sec:KTACA}). These approaches leverage two ideas: client auditing and anonymity. We first explore how to use these ideas on their own to solve the problem, designing one defense that uses each idea (KTCA and AKM). We then combine the ideas into a third design that overcomes some of the limitations of the first two approaches (KTACA). 
The detection is probabilistic---detection is not immediate, but as time passes, it becomes improbable for the attacker to avoid detection.
For all three defenses our design goals include (1) avoiding reliance on third-parties for auditing, (2) using existing infrastructure where possible for simpler deployability, and (3) leaving existing secure messaging protocols largely unchanged, with only small extensions.

The first idea we build on, client auditing, is based on key transparency (\eg CONIKS~\cite{melara2015coniks}), an auditing approach similar to Certificate Transparency \cite{scheitle2018rise}. With key transparency, a service provider maintains a write-only log of public keys for each client device. Auditors detect when a provider equivocates by advertising different logs to different users. Typically, auditors are assumed to be well-connected servers, either run by third parties or service providers that collaboratively audit each other. We explore client auditing, which uses secure messaging clients as auditors rather than having dedicated third-party servers or service providers perform this function. Client auditing is necessary because there is no evidence that third parties would be willing to perform this function, nor that secure messaging providers would work together to audit each other.

The second idea we build on, anonymity, enables a device to access a service without revealing its identity to the server. Anonymity is helpful because current secure messaging providers know the identity of clients when they request public keys; the provider distributes public keys for each user to bootstrap secure connections. We explore anonymous key monitoring to make it difficult for the server to deliver fake keys to specific users while avoiding detection. We are not interested in anonymous communication among clients because users of secure messaging systems typically want to communicate openly with people they know. Instead, we explore a more limited notion of anonymity, which requires hiding the device identity from a service provider.

We combine these ideas into a third design (KTACA) that uses anonymous client auditing. With the first client auditing approach, clients must exchange auditing information because the service provider, knowing the identity of the clients, could otherwise equivocate by sending different information to each client.
Anonymous client auditing allows each client to avoid the overhead of exchanging auditing information. Since the service provider can't identify the clients, it can't equivocate without a high probability of being caught. Thus repeated queries are sufficient to detect equivocation.

Our contributions include:

\begin{enumerate}[topsep=1pt,noitemsep,leftmargin=1.5em]
\item A detailed description of fake key attacks and an implementation demonstrating the feasibility of the attacks. 
\item The design of three automated fake key defenses, one that uses client auditing (KTCA), a second that uses anonymous key monitoring (AKM), and a third that uses anonymous client auditing (KTACA), along with their advantages and limitations. These defenses are oblivious to users, freeing them of the responsibility to protect themselves against key attacks through manual key verification.
\item A formal security analysis of the defenses that explains which attacks they can detect.
\item An analysis of an implementation of the defenses to explore their performance and feasibility.
\item A comparison of the defenses and a discussion of their trade-offs, as well as implications for this line of research. 
\item A taxonomy of MITM and impersonation attacks (see Fig 1). Our analysis shows that key transparency detects some impersonation attacks without proof of the attack. This limitation has not been discussed previously.
\end{enumerate}

Service providers have an incentive to deploy automatic detection to protect their users since the primary goal of their service is to provide private communication.
Detecting attacks also protects the service provider's reputation.
Likewise, automatic detection is a strong deterrent for attackers attempting an attack
and for service providers acting maliciously (\eg responding to a government subpoena).
Finally, the defenses increase usability since the effort to manually verify a key can be limited to many fewer situations.

\section{Background and Related Work}

We first provide background on secure messaging applications and then discuss related
work that seeks to help users verify public keys when using secure messaging applications.

\subsection{Secure Messaging Applications}

Secure messaging applications use many different protocols to provide end-to-end encryption. One family of applications is based on the Signal protocol~\cite{signalProtocol, cohn2017formal}, hereafter referred to as Signal. These include the Signal app, WhatsApp~\cite{whatsappencryption}, Facebook Messenger~\cite{messengersecret}, Skype~\cite{skypeprivate}, and Riot\footnote{Riot uses Olm, an implementation of the Signal Double Ratchet algorithm, for one-to-one encrypted communication (https://gitlab.matrix.org/matrix-org/olm/blob/master/docs/olm.md).}, all of which directly use the Signal protocol, as well as Wire~\cite{wiresecurity}, and Viber~\cite{viberencryption}, which use their own implementation but follow the same concepts.
Another family of applications (\eg iMessage, Threema, and Wickr) uses a proprietary protocol that bootstraps encryption by exchanging public keys using a central server, similar to the initialization used by the Signal protocol. 

Our work applies to all of these systems since they all use a central server to exchange public keys and route messages between users. Our focus is on ensuring that the public
keys exchanged through the central server are verified as authentic, rather than fake keys substituted by
an attacker. Most of these apps (except iMessage and Skype) use some manual system to verify keys,
and all could use an automated system such as those we describe.
Our work may also apply to Telegram, which uses a proprietary protocol based on Diffie-Hellman, with messages exchanged through a central server. Telegram also includes a method to authenticate the
exchanged Diffie-Hellman parameters, which would likewise benefit from automation.

Chase et al.~\cite{chase2020signal} use anonymous authentication to authorize changes to an encrypted membership list stored on a messaging server.  Two of our defenses (Section~\ref{Subsec:AKR} and Section~\ref{sec:KTACA}) require anonymous queries to hide the identity of the requestor, which is different from anonymous authentication.

\subsection{Verifying Keys}
\label{related_work_verifying_keys}
Secure messaging applications often contain a method for users to verify each others' public keys, such as scanning a QR code from each others' phones if they are co-located or reading their key fingerprints over a voice call.
This process has been called an {\em authentication ceremony}, and typical messenger designs
only prompt users to perform it when their public keys change.

Prior research shows that users do not understand the need for the authentication ceremony and find it difficult to perform.
Schroder et al. \cite{schroder2016signal} demonstrated that most Signal users failed to correctly verify their conversation partner's key due to usability issues and an incomplete mental model.
Herzberg and Liebowitz~\cite{herzberg2016can} conducted a laboratory user study that provided high-level information about the risks of secure communication. Only 13\% of the users could complete the authentication ceremony successfully.
Similarly, Vaziripour et al. ~\cite{vaziripour2017you} conducted a laboratory study where pairs of participants received high-level instructions to make sure they were communicating with the person they intended. Only 14\% of the participants completed the ceremony.
These studies indicate that users don't understand the risk of a MITM attack when using secure messaging applications, do not understand that the authentication ceremony helps them thwart an attack, and have difficulty finding and completing the authentication ceremony.

Several researchers have recently designed and evaluated improvements to the authentication ceremony interface.
Vaziripour et al.~\cite{vaziripour2018action} modified the Signal application UI to encourage users to perform the authentication ceremony and made the ceremony easier to find and use.
They reported that 90\% of the participants could find and complete the ceremony using the redesigned version of Signal.
Even with these improvements, it is still a burden and unrealistic to expect users to perform the ceremony all the time. Wu et al.~\cite{wu2019something} instead redesigned the Signal application UI based on risk communication principles that help the user decide whether to perform the authentication ceremony, taking into account risk likelihood and severity, response efficacy, and cost. This approach showed improvements in user understanding of the ceremony and the ability of users to make decisions based on their judgment of these factors.
 
Our work seeks to automatically verify public keys, relieving the burden of an authentication ceremony on users by distinguishing between legitimate key changes and attacks. One approach in this direction uses social media accounts to provide additional channels for verifying keys. Keybase is a key directory that links a user's social media accounts to their encryption keys to increase confidence that a received public key belongs to the right person.
Vaziripour et al.~\cite{vaziripour2019bother} semi-automated the authentication process in the Signal application using social media accounts, similar to Keybase.
They found that automating the authentication ceremony and distributing trust with additional service providers is promising. However, users were skeptical of using social media accounts due to a lack of trust, and thus recommended that more trustworthy third-party actors are needed for this role. We also note that using social media accounts requires users to have these accounts. In some cases, social media companies could collude with the messaging application (such as Facebook owning WhatsApp). Moreover, this approach puts the onus on the user to manually verify that the social media account owner is the person they are trying to authenticate.

A candidate for automating the ceremony is Key Transparency (KT).
CONIKS~\cite{melara2015coniks} (and~\cite{chaseseemless} that builds on it) describes a KT system designed for secure messaging that partly inspired Google's Key Transparency project\footnote{Refer \url{https://github.com/google/keytransparency}}.
The CONIKS architecture has each secure messaging provider maintain a public ledger of their user's keys, and providers audit each other to detect equivocation. It is assumed that clients can communicate with each other independently from their provider, enabling them to contact the auditors out-of-band and to alert other parties if any equivocation is detected. Today's secure messaging providers have not implemented ledgers or auditing, nor do they inter-operate in the way CONIKS envisions. We propose a defense KTCA using KT wherein the clients act as auditors instead of the service providers. 
KT is analogous to Certificate Transparency (CT) \cite{scheitle2018rise}, an approach for detecting fake  certificates. 
In both CONIKS and KTCA, clients audit their service provider to verify their key is in the transparency log. However, they differ in how they audit the server for equivocation\textemdash detecting when the server advertises a different log to different clients.
CONIKS assumes multiple non-colluding providers publish each others’ STRs, and clients download STRs from multiple providers and compare them for equivocation. KTCA does not require multiple providers since we believe it is infeasible given the current lack of cooperation (and lack of incentives for cooperation) among messaging providers. KTCA shows that single-provider detection of equivocation is feasible, and our analysis shows its limitations. Also, KTCA is customized for messaging apps, such as not needing a proof-of-absence as CONIKS provides, leading to some differences in the Merkle Tree design.

Unger et al.~\cite{unger2015sok} produced a systemization of knowledge (SoK) for secure messaging and considered security, usability, and adoption properties. They evaluate many trust establishment approaches, including centralized key directory systems like Signal. They evaluated whether approaches could prevent or detect operator (\ie service provider) MITM attacks and noted the key directories are vulnerable to attack. Their evaluation included the addition of auditable logs like CONIKS to a key directory to detect attacks. Our work builds on their evaluation in several ways. First, we provide an in-depth discussion on how to adopt CONIKS to a single operator. We also provide two additional approaches to detect attacks automatically. Their SoK focuses entirely on manual approaches to comparing fingerprints.

\section{System model}
\label{sec_threat_model}

There are two entities involved in the messaging system, (1) client messaging applications
on users' devices, and (2) the server.  
Clients can communicate with each other (1) via messages in the secure messaging application that are encrypted and routed through the server, and (2) via messages that are sent directly between the clients without being routed through the server (via some out-of-band communication channel). 
Whenever the secure messaging app is installed on the client, a public key is generated and published on the server. A new key can be associated with a phone number whenever the app is re-installed on the same device or another device due to a device upgrade, lost device, etc. 
A client can become disconnected from the system when it loses network access.

The server relays application-related messages between the clients and distributes public keys.
We refer to messages between clients that flow through the server as {\em in-band}, and messages between clients that flow directly without going through the server as {\em out-of-band}.
The server knows the sender and recipients' identity (phone number, public key, IP address) for all communications that it relays between clients. 
The server also knows the contact list for each client.

We assume perfectly synchronized clocks and bounded delays for all communication. 
Time is divided into regular intervals or epochs, such as an hour or a day.

\section{Adversary Model}

The adversary \A controls the server, giving it control over public key distribution and access to all encrypted messages plus metadata that flow through the server.  
Also, the adversary is an active global attacker that can launch a MITM attack against any insecure Internet communication globally to read, modify, inject, and block messages.

The adversary has no control over a user's client application. Because the cryptographic keys are generated at the client, and public keys are stored on the server for distribution to other users, the adversary has access only to public keys. We exclude a compromised client from our threat model because attacks on the client can be mitigated through an open-source client app that the provider does not control, which helps prevent the attack. Also, an attack on the client app can be detected by auditing since many researchers have access to the app and can inspect it.

The adversary does not attack the human user, so social engineering attacks are outside the scope of the paper.

\subsection{Adversary type and goals} 
The adversary could be (1) law enforcement or an oppressive regime coercing the server to conduct an attack, or (2) hackers compromising the server to conduct an attack. 
The adversary's goal is to compromise message confidentiality and integrity.
The adversary may be interested in compromising communication between a specific pair of users (Alice and Bob), or in compromising communication between one user (Alice) and all of the users with whom Alice communicates (Alice's contacts). 

\subsection{Description of Attacks} 
\label{sec_attack_desc}
\A generates fake public/private key pairs and distributes the fake public keys, instead of the original users' public keys, to conduct attacks on some conversations. Throughout the paper, we refer to this as a fake key attack, referring to a fake long-term public key of a user. Once \A has successfully distributed a fake key, it can generate any ephemeral symmetric keys needed to have read and write access to all messages in a conversation. If \A is able to complete a fake key attack, there are two specific attacks it can launch in existing systems that do not employ our defenses:

\paragraph{Man In The Middle Attack (MITM)} \A can launch a MITM attack against Alice and Bob by impersonating Alice to Bob and simultaneously impersonating Bob to Alice. \A can read, modify, and inject messages into the conversation.

To conduct this attack on a new connection, \A first generates two fake public keys. When Alice attempts to create a secure connection with Bob by retrieving Bob's public key, \A suppresses Bob's key and presents one of the fake keys (as Bob's key) to Alice. When Alice sends the first message to Bob containing her public key, \A replaces Alice's key with the other fake key.
The same pattern occurs for Bob if he initiates the conversation with Alice. \A can also launch a MITM attack against clients with existing connections by sending fake key updates to both Alice and Bob.

\paragraph{Impersonation Attack} 
\A launches an impersonation attack against Alice by either impersonating {\em as Alice} to her contact or impersonating a contact {\em to Alice}. To impersonate as Alice, \A provides a fake key for Alice to her contact. To impersonate to Alice, \A provides a fake key for her contact to Alice.
As an impersonator, \A can be either the initiator or the respondent of a conversation.
Alice is oblivious to the impersonation attack and is unable to communicate with her contact.

This attack applies to both new and existing secure connections. 
For new connections, \A distributes a fake key as the conversation begins. For existing connections, \A sends out a fake key update. The details depend on whether \A is impersonating to Alice or as Alice, and whether \A is the initiator or respondent in the initial communication.

The impersonation and MITM attacks just described illustrate a \textit{pair-targeted attack} where \A attacks a single pair of participants (Alice and Bob) as shown in Fig.~\ref{fig:ThreatModel} parts (a) and (c). 
A variation of these attacks is a \textit{client-targeted attack} where \A launches the attack on Alice and \textit{some or all} of her contacts as shown in Fig.~\ref{fig:ThreatModel} parts (b) and (d).

\begin{figure}
	\centering
	\includegraphics[width=0.45\textwidth]{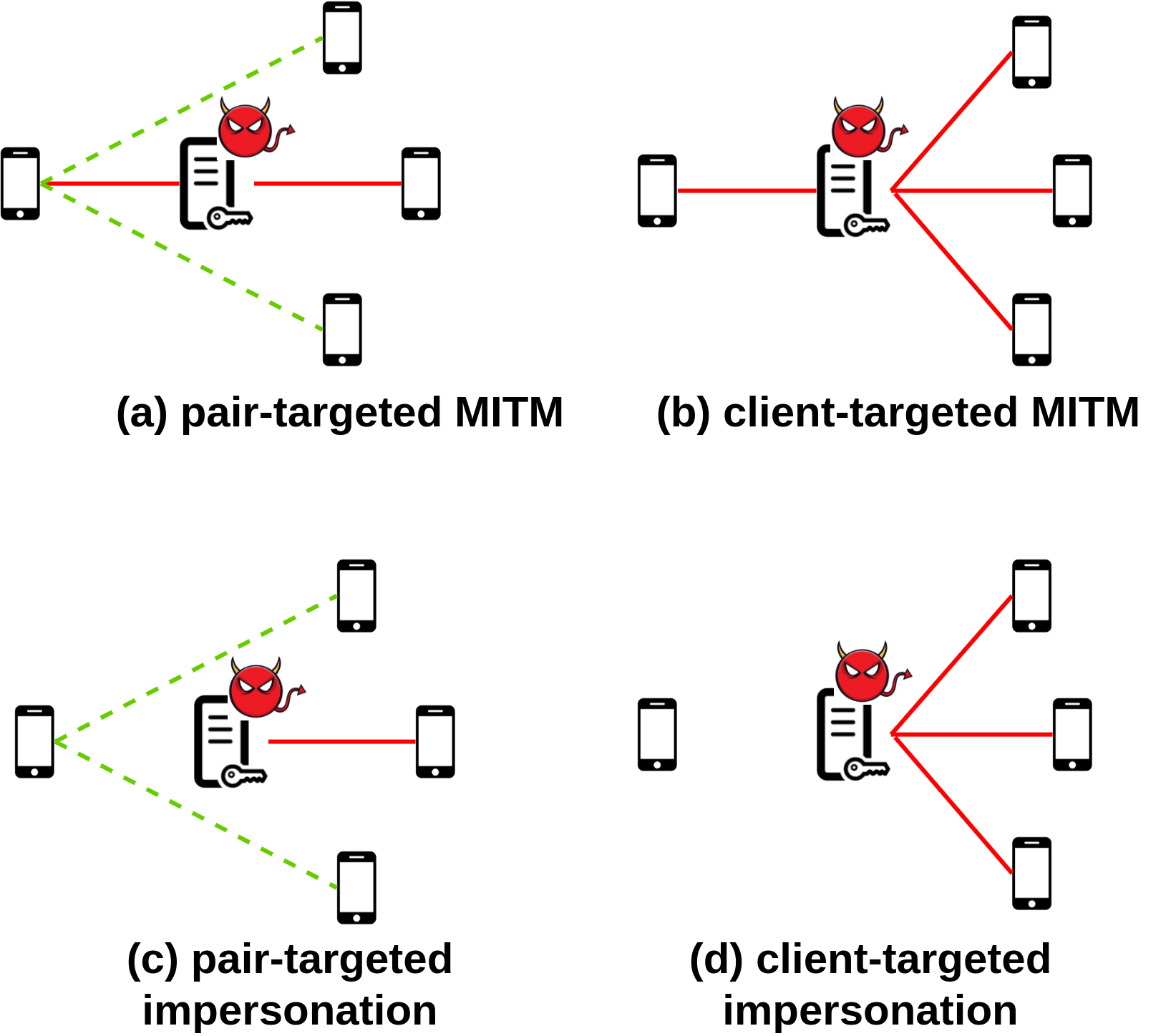}
	\caption{Four types of fake key attacks. Each diagram represents Alice (left), her contacts (right), and an adversary (center).
	The green lines represent secure connections; the red lines represent compromised connections.
    Note that in each situation, there may be other pairs of clients communicating securely that are not shown.}
	\label{fig:ThreatModel}
\end{figure}

\paragraph{\textbf{Roadmap}}
The rest of the paper is organized as follows: 
We present three defenses---KTCA, AKM, and KTACA---in Sections~\ref{key_transparency_defense}, ~\ref{Subsec:AKR}, and ~\ref{sec:KTACA}, respectively. KTCA has stronger assumptions than AKM, but it has better security properties than AKM. We design KTACA to combine the strengths of KTCA and AKM. Table~\ref{tab:comparison} compares  how these defenses perform against attacks described in Section~\ref{sec_attack_desc}. Section~\ref{subsec:KeyHist} presents short-lived attack monitoring, which is used in all of our main defenses to detect an adversary that launches a fake key attack and quickly restores the correct key to avoid detection. We present a security analysis of the defenses in their respective sections and a performance analysis in Section~\ref{sec:proof_of_concept_performance}.

\begin{table*}
	\centering
	\caption{Comparison of Defenses showing who detects/prevents each attack}
	\footnotesize
	\label{tab:comparison}
	\begin{tabular}{l|ccc}
	\toprule
		\textbf{Attack} & \textbf{KTCA} & \textbf{AKM} & \textbf{KTACA}  \\
		\midrule
		
		
		
		
		
		
	
	
		Pair-targeted MITM & All clients w/ PoM & All victim clients &  All victim clients w/ PoM  \\
		\midrule
		Client-targeted MITM & All clients w/ PoM (equivocation)&  &  All victim clients w/ PoM (equivocation)  \\
		& All victim clients (non-equivocation) & All victim clients & All victim clients (non-equivocation)\\
		\midrule
		Pair-targeted impersonation & All clients w/ PoM & All victim clients & victim client w/ PoM\textsuperscript{\textdagger}   \\
		\midrule
		Client-targeted impersonation & All clients w/ PoM (equivocation)&  &  All victim clients w/ PoM (equivocation)  \\
		& victim client (non-equivocation) $^*$&victim client$^*$ & victim client (non-equivocation)$^*$\\
\bottomrule
	  \end{tabular}

	  \vspace{2mm}
	*~= The victim client whose fake key is distributed detects the attack. \hspace{1cm}
	\textdagger~= The victim client who receives a fake key detects the attack.
	
\end{table*}

\section{Key Transparency with Client Auditors (KTCA)}
\label{key_transparency_defense}

Our first approach to defending against fake key attacks is an adaptation of key transparency that relies on secure messaging clients to audit the server instead of the usual approach of relying on multiple service providers to audit each other (\eg CONIKS and Google Key Transparency).  
We believe client auditors are more realistic to deploy since providers do not currently cooperate, and there is no other non-colluding auditing infrastructure available. 

Key transparency leverages a transparency log of public keys to detect when \A advertises a fake public key for a client. It also detects if \A attempts to avoid detection by equivocating, \ie advertising different logs to different clients. The secure messaging clients perform the monitoring. Some clients can be offline during some epochs; therefore, different clients will be available to monitor keys and check for equivocation in every epoch.

\subsection {Definitions and Assumptions}
We say that two entities have a $\delta$-connection in a given epoch if every message sent by one directly to the other is received with maximal delay $\delta$.
We say a client is \textit{benign} if it has the correct KTCA implementation.

In every epoch $e$, let $G_e$ be a graph whose nodes are the benign clients and whose edges are pairs of benign nodes that are contacts of each other. Because the contacts are connected through the server, they have a delay of $2\cdot\delta$. Notably, the server has $\delta$ connections with all clients in $G_e$. 

Due to clients going offline/online during epochs, the graphs $G_{e}$ and $G_{e+1}$ likely will not have the same set of nodes. KTCA assumes at least one overlapping benign client across two consecutive epochs. 
Having at least $N/2$ benign clients online every epoch, where $N$ is the total number of benign clients, guarantees an overlap of at least one client across two consecutive epochs.
Currently, messaging applications such as WhatsApp have $70\%$ of clients active every day~\cite{whatsappStats}.

We also assume  $2\cdot(diam(G_e)+1)\cdot\delta < len(e)$, where $diam(G_e)$ is the maximum number of edges for a shortest path between any two clients in $G_e$. 
This assumption accounts for the time it takes for a client to retrieve information from the server and share it with all of the other clients in the network. 
If any client sends a message to its contacts, who then repeatedly relay the message to their contacts, the message will reach all clients in $G_e$ in one epoch.
This assumption is reasonable because most delays on the Internet are short---rarely over 200 ms. 

\subsection {Design}
  
    The server maintains a Merkle binary prefix tree of all the registered client's public keys, inspired by CONIKS~\cite{melara2015coniks}. 
    Each node corresponds to a unique prefix $i$. Each branch of the tree appends a 0 (left child) or a 1 (right child) to the parent's prefix.
    
    
    A privacy goal for a Merkle tree is that an attacker cannot use the Merkle tree data that proves user $i$ owns key $k$ to determine whether another user $j$ exists in the tree. 

    There are three node types in the tree created using a collision-resistant hash function\footnote{In the exposition, we use a keyless hash $H$ for simplicity, which suffices under the Random Oracle Model. For security in the standard model, the protocol should be interpreted as using a keyed hash function $h_k$, where $k$ is a public random string.} $H()$. 
    Leaf nodes are defined as:
    \begin{equation*}
        h_{leaf} =  H(k_{leaf}||i_{client}||\ell|| H(client,public\_key_{client}))\\
    \end{equation*}
    where $k_{leaf}$ is a leaf-specific nonce, $i_{client}$ is the index for a $client$, $\ell$ is the depth of the leaf node in the tree, and $client$ is a unique identifier (\ie phone number) for the user.
    
    For a leaf node, index $i_{client}$  is a hash of the client's identifier.
    \begin{equation*}
        i_{client} =  H(client)
    \end{equation*}
    
    Interior nodes are defined using their two children as: 
     \begin{equation*}
        h_{interior}=H(k_{interior}|| h_{child.0}||h_{child.1}|| i_{interior} || \ell )
    \end{equation*}
    where $k_{interior}$ is an interior-specific nonce, $i_{interior}$ is the prefix for the interior node, $\ell$ is the depth of the interior node in the tree. 
    
    The $k_{leaf}$ and $k_{interior}$ nonces ensure that the input strings to the hash function for leaf and interior nodes differ, and the outputs can only match if there is a collision. Some leaf nodes do not correspond to a registered client and are simply random values.

    To optimize space in the tree, the leaf node corresponding to a registered $client$ could be placed at a depth $\ell$ where the first $\ell$ bits of $i_{client}$ are a unique prefix in the tree.
    This optimization allows an offline attacker to infer other potential clients nearby in the tree. 
    CONIKS incorporated a verifiable unpredictable function in its design to eliminate this privacy leak. 
    The function can be verified with the public key but only computed using the private key. This prevents an offline attacker from guessing possible usernames and inferring that they might be present in nearby nodes.
    
    Rather than adopt this method from CONIKS, we use a different design that sacrifices some storage for reduced computation. To prevent the above privacy leak, we represent empty leaves as random values to be indistinguishable from non-empty leaves. This design choice means we cannot support proof of absence like CONIKS. 
    Proof of absence is not necessary for secure messaging apps because the messaging server is responsible for maintaining the tree, and every registered client must have a key in the tree. We also place each node for a $client$ at a depth corresponding to the first $\ell+r$ bits of $i_{client}$, where $\ell$ is the first $\ell$ bits of index $i_{client}$ that are a unique prefix in the tree, and $r$ is chosen randomly (uniformly) between $1$ and $\ell$ by the server.
    The extra overhead is at most twice the number of hashes to a node in the tree. The design ensures that a sibling leaf for a non-empty leaf is always either an interior node or a random-valued leaf node. An attacker cannot reliably infer whether another client with the same prefix exists in the tree. 

    At regular intervals (epochs), the server generates a \textit{\acrfull{str}} by signing the root of the Merkle tree, along with other metadata, such as a hash of the previous epoch's \acrshort{str} and epoch number.

A Merkle tree provides an efficient method for clients to verify that their key is included in the current tree without obtaining a full copy of the tree. 
The server provides a client with the current \acrshort{str} along with a \acrfull{poi} for their key.
The \acrshort{poi} contains the leaf node corresponding to the client's key and the hash of each sibling sub-tree along the path from the leaf node to the root of the Merkle tree. 
Given a tree with $N$ nodes, there are $log(N) - 1$ interior nodes along the path from the leaf to the root of the tree.
To verify the \acrshort{poi}, the client begins at the leaf, hashes the hash of the leaf along with the hash of the leaf of the sibling node to produce the value of the parent node. The process of computing the value of the parent node using the current contents of a node and the hash of its sibling node provided in the \acrshort{poi} continues until the root of the tree. The computed value for the root of the tree is verified using the most recent \acrshort{str}. 

\subsubsection {Client auditors}

The defense adds new \textit{auditing messages} that are exchanged between the clients to detect fake key attacks. These messages are in addition to the normal messages between clients that flow through the server and use all of the same security mechanisms. These messages represent the edges of graph $G_e$.

The following describes the key monitoring and server auditing performed by all clients. All clients perform the following: (1) monitor their own key, (2) monitor the keys for all their contacts, (3) audit the server to detect equivocation, and (4) perform short-lived attack monitoring as needed.

\begin{enumerate}[topsep=1pt,noitemsep,leftmargin=1.5em]
    \item To monitor their own key, clients request an \acrshort{str} and \emph{\acrfull{poi}} for their key from the server at the beginning of each epoch.
    Clients verify (a) that the server is publishing a linear history of \acrshort{str}s by confirming the previous \acrshort{str}'s hash is in the current \acrshort{str}, (b) that the \acrshort{str}'s signature is valid, and (c) that their public key is in the tree using \acrshort{poi}.
    If a client does not receive an \acrshort{str} or \acrshort{poi} within $2\cdot\delta$ time after the beginning of an epoch, or the \acrshort{str} is invalid, it considers this an attack, and the client disconnects from the server.
    If a client disconnects from the graph $G$, then it requests \acrshort{str}s and \acrshort{poi}s for all missed epochs when it comes online and re-joins the graph. Then it verifies the validity of all \acrshort{str}s and \acrshort{poi}s.
    
    \item A client monitors the keys of its contacts on each new key lookup or key update. The server gives client $j$ the key for client $i$ that is included in $STR_e$ and its corresponding \acrshort{poi}, which client $j$ can use to verify the key is in the current tree. If client $i$ updates the key after $STR_e$ is generated, then the server also sends client $j$ the updated key, which is not included in $STR_e$ but will be included in $STR_{e+1}$. 
    

    \item Clients audit the server using in-band auditing messages for equivocation by sending $\acrshort{str}_e$ to their contacts during each epoch, immediately upon receiving it. If a client detects two conflicting, signed \acrshort{str}s during an epoch, this constitutes a \emph{\acrfull{pom}}.
    Clients detecting or receiving a \acrshort{pom} immediately relay the \acrshort{pom} to their contacts so that all clients in the network quickly obtain the \acrshort{pom}.
    Clients that do not receive $\acrshort{str}_e$ from the server during a given epoch instead relay the first \acrshort{str} they receive from their contacts.
    
    \item On every key update a client receives, it performs \textit{short-lived attack monitoring} (\S\ref{subsec:KeyHist}) to detect attacks where \A quickly restores a correct key. \textit{Short-lived attack monitoring} detects the attack where \A provides a fake key in epoch $e$ and restores the key before epoch $e+1$, so the fake key never appears in the subsequent \acrshort{str}.
    

\end{enumerate}

\subsection {Analysis}

If \A does not equivocate, then every client in the graph receives the same \acrshort{str} during every epoch. Each client can verify their key is in the \acrshort{str} using the \acrshort{poi}.

If \A equivocates, then at least one pair of clients has conflicting \acrshort{str}s. Namely, at least one edge in the connected graph has a $2\cdot\delta$ connection between two clients that receive different \acrshort{str}s. The clients on each edge with different \acrshort{str}s will detect the equivocation and forward the \acrshort{pom} to their contacts.

The detection of conflicting \acrshort{str}s by a client proves that at least one fake key exists in the system. However, the client cannot determine which keys are fake, only that the server has equivocated.
The owner of a key is the only one that can confirm their correct key is present in an \acrshort{str}. 
Other clients can confirm that the key exists in the \acrshort{str}, but they cannot confirm it is the correct key.


In all cases (a)-(d) in Fig.~\ref{fig:ThreatModel}, if \A attacks by equivocating and gives different STRs to clients connected to the compromised edges, the attack is detected with a \acrshort{pom} by all clients that are connected to graph $G_e$ and exchange their STRs. It should be noted that \A could attack an edge $e$ by partitioning the graph $G$ into subgraphs by either (1) disconnecting multiple edges or (2) executing MITM on multiple edges in the same epoch in a way that the clients of edge $e$ come in different subgraphs $G_1$ and $G_2$. In this case, attacks on edges within the subgraphs $G_1$ and $G_2$ are detected with PoM, but attacks on edges connecting the two subgraphs cannot be detected quickly. To prevent this attack's detection, \A has to permanently disconnect these connections or execute a MITM for all connections between any two clients across these subgraphs. 

Suppose two clients from different subgraphs form a secure connection that connects graph $G$. In that case, the attack is eventually detected with a \acrshort{pom}, which is a strong deterrent for \A executing this attack.
If \A does not partition the graph, the attacks are detected within one epoch as described in Theorem 5.2.

To illustrate the graph partitioning attack, consider the case where \A conducts a client-targeted MITM attack. \A creates a fake key for Bob and delivers it to all of Bob's contacts. To avoid detection during auditing, \A sends Bob an STR generated with his key, but all other clients receive a second STR generated with the fake key. To avoid detection, \A also modifies all of Bob's outgoing auditing messages (since it is acting as a MITM), and replaces the first STR with the second STR, which matches the STR that all the other clients have. For the incoming STRs, \A replaces the second STR. In this case, \A conducts a successful MITM even though Bob and his contacts verify that their conversations are secure. As noted above, \A must continue this attack indefinitely to avoid detection and expand it to include any new clients that contact Bob.

In client-targeted attacks on Bob, \A could hand out fake keys and not equivocate. For a client-targeted MITM attack, all the victim clients detect the attack. For a client-targeted impersonation attack, only Bob detects the attack. However, the impersonation attack can continue since there is no way to notify Bob's contacts that they are under attack automatically. Bob has to notify his contacts manually to make them aware of the attack because \A isolates Bob from communicating with anyone via the network. There is no \acrshort{pom} in either case. 
If \A refuses to hand out an STR, a client detects the attack without any \acrshort{pom}.

\label{section:KTCA_detailed_analysis}
\begin{lemma}\label{lem:KTCA1}
\textit{PoI-Lemma:} For any epoch $e$, it is not feasible for \A to generate an \acrshort{str} and two \acrshort{poi}s that  prove the inclusion of two different public keys in the \acrshort{str} for the same $client$.
\end{lemma}

\begin{proof}

There are two methods for \A to distribute an STR with two \acrshort{poi}s where one proves a client owns $public\_key_{client}$ and the other proves the client owns $public\_key_{fake}$. 
    
Method 1: The \acrshort{poi}s map the conflicting public keys to the same leaf node in the Merkle tree. 
    \begin{equation*}
        \begin{aligned}
        h_{leaf} = H(k_{leaf}||i_{client}||\ell|| H(client,public\_key_{client})) \\
        = H(k_{leaf}||i_{client}||\ell|| H(client,public\_key_{fake}))
        \end{aligned}
    \end{equation*}
If $public\_key_{fake}$ and $public\_key_{client}$ are different, then the outputs of the inner collision-resistant hash function must differ. Thus, the inputs to the outer hash function $H$ differ, and the outputs can only be equivalent if a collision occurs.

Method 2: The first \acrshort{poi} maps the real public key to a leaf $h_{leaf}$ at depth $\ell$ in the Merkle tree. The second \acrshort{poi} maps the fake public key to a fake leaf $h_{fakeleaf}$ that equals an $h_{interior}$ node at depth $\ell^\prime$ on the path to $client$ where $\ell^\prime$ is less than $\ell$. If they are equal, then the fake key validates as a leaf node in the Merkle tree.
    \begin{equation*}
        \begin{gathered}
        h_{interior} = H(k_{interior}||h_{child.0}||h_{child.1}|| i_{interior} || \ell^\prime ) \\
        h_{fakeleaf} = H(k_{leaf}||i_{client}||\ell^\prime|| H(client,public\_key_{fake}))
        \end{gathered}
    \end{equation*}
The inputs to the hash function $H$ differ since, by definition, the input strings begin with a different nonce $k_{interior} $ and $k_{leaf}$. The outputs can only be equivalent if a collision occurs.

Both methods depend on a collision for $H$. Therefore, if it is feasible for \A to come up with an \acrshort{str} and two valid \acrshort{poi}s, corresponding to two different public keys for the same client, then this implies that it is feasible to find collisions in $H$, contradicting the assumption that $H$ is a collision-resistant hash function. 

\end{proof}


In Appendix section~\ref{sec:kt_proofs_appendix}, theorem~\ref{thm:KTCA1} considers clients to be online during every epoch and a key update happens at a fixed time during an epoch; theorem~\ref{thm:KTCA2} assumes only fixed-time key updates. 
The following theorem considers that clients can go offline during epochs and update their keys anytime during an epoch. From our definition, $G_e$ is a graph whose nodes are the benign
clients and whose edges are pairs of benign nodes that are contacts
of each other.

\begin{theorem} \label{thm:KTCA3}
Let $j$ be a benign client connected in epochs $e$ and $e'$ where $e<e'$, and $i$ be a benign client connected in epochs $e^*$ and $e"$ where $e^* < e < e"$, and let $t_e$ denote the time that epoch $e$ begins. Let $e_m = \max(e', e")$. Assume that $j$ receives $i$'s key from the server at time $t'$ within the epoch $e$. Then one of the following holds:

\begin{enumerate}
\item Client $i$ detects that the server is corrupt, at or before $t_{e"}+2\cdot\delta$.
\item  All clients in $G_{e_m}$ detect that the server is corrupt, with PoM, at or before $t_{e_m}+2\cdot(diam(G_{e_m}+1))\cdot \delta$.

\item Client $j$ detects that the server is corrupt at or before $t_{e'}+2\cdot\delta$.

\item Client $j$ receives $i$'s public key during epoch $e$.
\end{enumerate}
\end{theorem}

\begin{proof}
In KTCA, we allow clients to update their keys any time during the epoch, and the updated keys are included in the Merkle tree in the next epoch. So, the detection process following a key update starts at the beginning of the next epoch following the update.

The first case describes when \A tries to create a fake key for a client without equivocation. If client $i$ in $G_{e"}$ does not receive \acrshort{str}s with \acrshort{poi}s for epochs $e^*+1$ to $e"$ 
for its own key within $2\cdot\delta$ time or receives any invalid \acrshort{poi} for a corresponding \acrshort{str}s, then client $i$ considers it an attack. Client $i$ detects this corrupt behavior at or before $t_{e"} + 2\cdot\delta$.


The second case describes equivocation detection. From our assumption, each epoch has more than $N/2$ clients online, ensuring that there is an overlap of at least one client in $G_{e'}$ and $G_{e"}$. 
Note that $e'$ and $e"$ can be in any order, and it is beneficial for \A to always equivocate at the later epoch. Let $e_m = \max(e',e")$ and $c_m$ be either client $i$ or $j$ who is online in epoch $e_m$. To equivocate, in epoch $e_m$ \A gives a different \acrshort{str} to the overlapping clients than it gives to client $c_m$.


If client $c_m$ receives a different \acrshort{str} than the overlapping clients,  then one or more clients in the graph receive $STR_1$, and one or more clients in the graph receive $STR_2$. As the graph $G_{e_m}$ consists of $2\cdot\delta$-connections between clients, there must exist at least one edge connecting two clients where one client receives $STR_1$, and the other client receives $STR_2$. Once these two clients exchange the \acrshort{str}s, a \acrshort{pom} exists. Then the \acrshort{pom} is forwarded to all clients in $G_{e_m}$ and is later forwarded to other clients as they connect.

The maximum time that is incurred for detection by all clients is the case when the two farthest clients (the ones that are $diam(G_{e_m})$ apart) receive the conflicting \acrshort{str}s and all the intermediate clients have not received an \acrshort{str} from \A.
The detection time is computed as follows. At the beginning of an epoch, within $2\cdot\delta$ time, client $x$ receives $STR_1$, client $y$ ($diam(G_{e_m})$ apart from $x$) receives $STR_2$ and all the intermediate clients do not receive \acrshort{str}s. Clients $x$ and $y$ send their \acrshort{str}s to their neighbors and in turn they also forward the \acrshort{str}s to their neighbors. Within $diam(G_{e_m})\cdot\delta$ time the node which is $diam(G_{e_m})/2$ apart from both $x$ and $y$ receives both the \acrshort{str}s and generates the PoM.
Eventually, in $diam(G_{e_m})\cdot\delta$ time this PoM is propagated within the whole $G_{e_m}$ network and all the clients in $G_{e_m}$ become aware of the equivocation including $x$ and $y$. Thus, to detect equivocation the maximum time is $2\cdot\delta + diam(G_{e_m})\cdot \delta +  diam(G_{e_m})\cdot \delta$, which is $2\cdot(diam(G_e)+1 )\cdot\delta$. The equivocation is detected at or before $t_{e_m} + 2\cdot(diam(G_{e_m}+1))\cdot \delta$.


Third, from the KTCA design, if any client $j$ asks for client $i$'s key, client $j$ must receive the key of $i$ and \acrshort{poi} within $2\cdot\delta$ time. However, as described in the first paragraph of the proof, \A can say that it will include the updated key in the next epoch. So, if a client does not receive a \acrshort{poi} or it is invalid for epoch $e'$, client $j$ considers it an attack and detects it at or before $t_{e'}+ 2\cdot\delta$. 

If the other cases do not hold, then client $i$ receives a valid \acrshort{str} with \acrshort{poi} for its own key, and all clients receive the same \acrshort{str}s in epoch $e"$. PoI-Lemma in KT ensures that if a valid \acrshort{poi} is in the \acrshort{str} that is consistent with the \acrshort{str} $i$ has, it is consistent with the key that client $i$ published in the tree. 
\end{proof}

\begin{theorem} \label{thm:KT4}
No false Proof of Misbehavior is ever created. If the server is never corrupt, no benign client will ever falsely detect that the server is corrupt.
\end{theorem}

\begin{proof}
By definition, a \acrshort{pom} occurs when a client receives two different \acrshort{str}s for an epoch. This equivocation cannot happen unless the server creates two different Merkle trees and signs two conflicting \acrshort{str}s. 
In a secure signature scheme,  only the signing key owner can generate a valid signature on an \acrshort{str}.
Also, a benign server always sends the \acrshort{str} and \acrshort{poi} on time.
\end{proof}

\section{Anonymous Key Monitoring (AKM)}
\label{Subsec:AKR}

The KTCA defense relies on clients to audit the server instead of trusting third-party auditors. However, it is vulnerable to graph partitioning attacks. We explore anonymous key monitoring to overcome this limitation.

Currently, secure messaging servers know the identity of the clients during public key distribution, which lends itself to being vulnerable to fake key attacks. 
This section explores AKM, a defense that leverages anonymous key requests from clients as they monitor keys to defend against fake key attacks.

Unlike KTCA, AKM successfully verifies keys even if \A partitions the graph $G$ and keeps it partitioned forever. Also, the graph's connectivity does not affect the detection time of key verification in AKM.

\subsection {Assumptions}

The attacker controls a server S and is also an eavesdropper on the communication entering and leaving the anonymous network. We assume an anonymous network with a (weak) anonymity property: when $n$ senders simultaneously send a short, fixed-length message to S through the anonymous network, the attacker can link the sender's identity $s_i$ to a received message $m_i$ with a probability at most $1/n + \epsilon(\kappa)$, where $\kappa$ is the security parameter, and $\epsilon$ is a negligible function of $\kappa$.

 We assume a maximum delay of $\delta$ when a client and the server communicate directly and a maximum delay of $\Delta>\delta$ when they communicate through the anonymous network. The duration of an epoch is much larger than $\Delta$.

Since AKM requires small, infrequent anonymous monitoring messages, it is possible to use strong-anonymous networks such as (1) mixnets (\eg  Nym{~\cite{diaz2021nym}}) and (2) Vuvuzela{~\cite{van2015vuvuzela}}.
The properties of AKM also make it feasible to rely on the weak anonymity guarantees from a system like Tor~\cite{torProject}. The weak anonymity guarantees suffice for AKM because:

\begin{itemize}[topsep=1pt,noitemsep,leftmargin=1.5em]
    \item Our messages are short, fixed-sized, and do not require high quality of service (such as low latency or a high-speed connection). Therefore, we can choose the proxies in a circuit uniformly from the list of available proxies (provided by the directory servers); this contrasts with the standard Tor client that chooses proxies in the circuit according to stability, latency, and available bandwidth.
    \item AKM has a client monitor their key over Tor once per epoch, a single request and response. The client constructs a new circuit each epoch to mitigate the risk of \A always controlling the entry and exit proxies and correlating the client's traffic.
    \item Finally, we introduce a randomized jitter by introducing a short randomized bounded delay before sending each request and randomizing the number of TOR proxies that participate in each circuit as described in ~\cite{gilad2013plug}.
\end{itemize}

\subsection {Design}
\label{sec:AKM_design}
The server distributes public keys to clients directly.
The server also supports an {\em \acrfull{akr}} for key monitoring.
An \acrshort{akr} is an unauthenticated public key request sent through an anonymous network to retrieve a key from the server. The server signs the public key in both responses (directly and through AKR).
An \acrshort{akr} prevents leakage of the client’s identity to the server by (1) not sending the client's identity or identifiable metadata in the request at the application layer, and (2) using a third-party anonymization service with a random delay for IP-layer anonymity of key requests.
AKR allows clients to retrieve their own keys or other clients' keys anonymously.

Adding support for an \acrshort{akr} should be a modest change to existing servers.
If a service does not support \acrshort{akr}s, two potential workarounds are for all clients to use the same generic identifier or to use a random identifier when making an \acrshort{akr}. We verified the feasibility of the first idea on a Signal server in the lab by retrieving keys anonymously using the same generic credential submitted by multiple clients in parallel TLS sessions.

To prevent a timing attack by \A, all clients make their AKR requests at the beginning of an epoch. Since \A controls the timing of each response, it could reply sequentially with sufficient delay to deanonymize the clients as a global passive adversary. However, since \A must commit to the response before returning it, the deanonymization occurs too late to help \A avoid detection.



To detect fake key attacks, all clients regularly perform two types of anonymous key monitoring plus short-lived attack monitoring.
\begin{enumerate}
    \item At the beginning of every epoch, Alice monitors her key using an \acrshort{akr} to ensure the server consistently distributes her key.  
    \item When Bob creates a new connection or receives a key update for a contact in epoch $i$, he monitors the contact's key using an \acrshort{akr} at the beginning of each epoch from epoch $i+1$ to $i+m$, where $m$ is the number of monitoring requests.
    \item Similar to KTCA, attacks where \A quickly restores a correct key can be prevented by performing \textit{short-lived attack monitoring}
    (\S\ref{subsec:KeyHist}). 
\end{enumerate}


At the beginning of an epoch, suppose Alice was to monitor her key and her contacts' keys in one bulk AKR. The server may be able to deanonymize her by comparing her request to a list of her new contacts and existing contacts that recently updated their key. To avoid this kind of deanonymization, Alice must create a fresh AKR that requests only one key for each key she monitors.

\subsection {Analysis}

Suppose \A performs a fake key attack on Alice by giving a fake key update to some of her contacts. 
To avoid detection, \A must reply to any \acrshort{akr} from those contacts with the same fake key. 
Simultaneously, \A must present the real key to Alice, who monitors her key during each epoch. 

For example, suppose Bob receives a key update for Alice containing a fake key. Bob then sends an \acrshort{akr} for Alice's key for the following $m$ epochs. To avoid detection, \A must return the same fake key for Alice to Bob during each epoch. Simultaneously, since Alice is sending an anonymous query for her key during each interval, \A must return Alice's correct key to her. The victim client, who is given a fake key in epoch $e$, detects the attack at or before $t_{e+m}$ with probability $1-(1/2)^m$, where $m$ is the number of epochs the victim client monitors the contact's key following a key update. Theorem 6.1 presents a detailed security analysis that considers clients going offline and \A giving a fake key update for a client to multiple contacts.

If a client $i$ does not receive a response for its \acrshort{akr} within $2\cdot\Delta$, where $\Delta$ is the maximum delay in the anonymous network between a client and the server, the client assumes the server is avoiding detection and considers it an attack.

Even though a client can detect an attack, AKM provides no \acrshort{pom}. AKM detects all four attacks presented in Fig.~\ref{fig:ThreatModel}.
If \A hands out fake keys and continues to answer all monitoring requests, the attack is detected 
because it is improbable for \A to deliver the correct key and the fake key in response to anonymous requests in a way that avoids detection. 
The attack is detected if \A fails to respond to an anonymous request or blocks any request.

In a client-targeted impersonation attack, \A can hand out a fake key for Bob and then answer all monitoring requests for Bob's key with the fake key. In this case, Bob is the only client that detects the attack. However, the impersonation attack can continue since there is no way to notify Bob's contacts that they are under attack automatically. Bob has to notify his contacts manually to make them aware of the attack because \A isolates Bob from communicating with anyone via the network.


Theorem 6.1 assumes that \A gives $i$'s fake key to (1) all new client connections and (2) some existing clients whom \A wants to attack in epoch $e$. This is a reasonable assumption since epochs are short (\eg an hour), and usually, there will be one or no new connections.
The assumption means that all monitoring requests not from $i$ expect the fake key, and only one request is for the real key.
Theorem A.4 (see Appendix~\ref{thm:AKM2}) assumes \A can partially target new connections in an epoch, making it more complicated for \A to avoid detection.

\label{section:AKM_detailed_analysis}
\begin{theorem} \label{thm:AKM}
Let $t_e$ denote the time that epoch $e$ begins. Assume $c$ clients receive a fake key for client $i$, during epoch $e$. 
Additionally, assume that the server cannot distinguish users with a probability significantly larger than 1/m (see Sect. 6.1 for a rationale). Then client $i$ detects that the server is corrupt at or before $t_{e+m}$ with probability $1-(1/(c+1))^m$, where $m$ is the number of epochs that the $c$ contacts monitor the client's key following a key update.
If either client $i$ or any of the $c$ clients disconnect during $m$ epochs, then the probability of detection is $1-\prod\limits_{i=1}^{m} max(1/(contacts_i+1), owner_i)$, where $contacts_i$ is the number of contacts online where ($contacts_i<c$), and $owner_i$ is a $1$ or $0$ depending on whether the owner is offline or online.
\end{theorem}
\begin{proof} 

When \A presents a fake key for Alice to $c$ of its contacts in an epoch and continues the attack for at least $m$ epochs, 1) the $c$ contacts monitor Alice's key for $m$ epochs, and 2) Alice monitors its key. All of these requests are indistinguishable from each other. To avoid fake key detection, \A has to deliver the correct key to Alice in every epoch and the fake key to its $c$ contacts.
\A knows it will receive $c+1$ requests, and only 1 of them should return the real key, and $c$ should return the fake key. So \A has c+1 choices for plausible ways to distribute the keys, and only one of them is correct. During each epoch, the probability of making the right choice is $(1/(c+1))$. So as the $c$ clients monitor Alice's key for $m$ epochs, the probability of \A making the right choice to avoid detection is $(1/(c+1))^m$. 

If Alice disconnects for any of those $m$ monitoring epochs, \A can distribute fake keys without detection during those epochs. Suppose at every epoch $i$ during the monitoring interval we have $contacts_i$ contacts online where ($contacts_i<c$), and $owner_i$ is a $1$ or $0$ depending on whether the owner is offline or online. If the owner is offline and \A knows this, \A can hand out fake keys reliably with probability 1. The probability of avoiding detection is the product of the probability during each epoch.  So the detection probability is:  $1-\prod\limits_{i=1}^{m} max(1/(contacts_i+1), owner_i)$.

If the attack is short-lived (less than $m$ epochs), and \A restores the original key, the attack is detected immediately using the \textit{short-lived attack monitoring} basic defense (described in~\S\ref{short-lived_proof}).
\end{proof}

As the number of epochs increases, the probability of \A avoiding detection (false negative) is negligible. Once the probability is negligible, Alice's contacts stop monitoring Alice's recent key change to avoid unnecessary bandwidth.

\section{Key Transparency with Anonymous Client Auditors (KTACA)}
\label{sec:KTACA}

While AKM defends against graph partitioning attacks, it does not provide a \acrshort{pom} and takes longer to detect an attack than KTCA. Thus, we next explore combining key transparency and anonymity to achieve the best of both approaches. 


KTACA relies on a key transparency log maintained by the server, with clients auditing the server anonymously.
Clients audit for equivocation by requesting STRs anonymously from the server instead of exchanging them with their contacts as in KTCA. 
Since the server does not know which client is requesting an STR, it makes it difficult for \A to equivocate and avoid detection.



\subsection {Design}

The server supports an {\em \acrfull{asr}} for STR monitoring.
An \acrshort{asr} is an unauthenticated STR request sent through an anonymous network to retrieve the STR for an epoch from the server. Similar to \acrshort{akr}, \acrshort{asr} prevents leakage of the client's identity to the server.
The server maintains a Merkle tree containing the public key of all registered clients, as in KTCA (\S\ref{key_transparency_defense}). 

\begin{enumerate}

 \item At the beginning of every epoch, all clients retrieve an \acrshort{str} and \acrfull{poi} for their key from the server. If a client is offline for some epochs, it retrieves the STRs and corresponding \acrshort{poi} for missed epochs.
    Clients verify (a) the server is publishing a linear history of \acrshort{str}s by confirming the previous \acrshort{str}'s hash is in the current \acrshort{str}, 
    (b) the \acrshort{str}'s signature is valid, and (c) that their public key is in the tree using \acrshort{poi}.
    If a client does not receive an \acrshort{str} or \acrshort{poi} within $2\cdot\Delta$ time after the beginning of an epoch, or the \acrshort{str} is invalid, it is considered an attack, and the client disconnects from the server.
    
    \item At the beginning of every epoch, all clients use an \acrshort{asr} 
    to anonymously retrieve an \acrshort{str} from the server. This \acrshort{str} should match the \acrshort{str} retrieved directly from the server in step 1. Otherwise, the client has conflicting \acrshort{str}s and detects an attack with a \acrlong{pom}. 
    Since all clients in the system request the STR anonymously, to equivocate the server has to correctly identify one STR request from millions of anonymous STR requests. 
    \item For each new key lookup or key update, the client receives a \acrfull{poi} for verifying the key is in the tree.
        \item Similar to KTCA and AKM, attacks where \A quickly restores a correct key can be prevented by performing \textit{short-lived attack monitoring}
    (\S\ref{subsec:KeyHist}). 
\end{enumerate}

In KTCA, clients audit for equivocation by comparing \acrshort{str}s among all of their contacts. With anonymous client auditing, clients leverage anonymity to ensure the server consistently delivers a linear history of \acrshort{str}s. 
Retrieving STRs through the anonymous network removes (1) the dependency on the assumption of a $\delta$-connected network, (2) the computation and network bandwidth for clients to exchange STRs to monitor for equivocation, and (3) the need to trust all the other clients in the network to participate in the monitoring process. 

\subsection {Analysis}

KTACA defends against the attacks presented in Fig.~\ref{fig:ThreatModel}.
For each attack, if \A equivocates, the attack is detected within one epoch.  The owner of the key and their contacts that receive a fake key obtain a \acrshort{pom}. Detection occurs because it is improbable for \A to deliver different STRs to only the victim clients in response to anonymous STR requests when all the clients in the system retrieve STRs anonymously in every epoch. 
Unlike KTCA, there is no mechanism to distribute the \acrshort{pom} to all the other clients in $G_e$. 
However, the mechanism used in KTCA to forward the  \acrshort{pom} to other clients can be easily added to KTACA. Note that the \acrshort{pom} will only be shared in the connected graph $G_e$ of which victim clients are part.


In client-targeted attacks, \A could hand out fake keys and not equivocate. For a client-targeted MITM attack, all the victim clients detect the attack. For a client-targeted impersonation attack, only the victim client that owns the fake key detects the attack. However, the impersonation attack can continue since there is no way to automatically notify the victim's contacts that they are under attack. The victim has to notify their contacts manually because \A isolates the victim from communicating with anyone via the network. There is no \acrshort{pom} in either case.

If \A refuses to hand out an STR, a client detects the attack without any \acrshort{pom}.

\begin{theorem}\label{thm:KTACA}
Let $j$ be a benign client that is online at epochs $e$ and $e'$ where $e<e'$, and $i$ be a benign client that is online at epochs $e^*$ and $e"$ where $e^* \le e < e"$, and let $t_e$ denote the time that epoch $e$ begins. Assume that client $j$ requests $i$'s key from the server at time $t'$ within the epoch $e$. Then one of the following holds:

\begin{enumerate}
\item Client $i$ detects that the server is corrupt, at or before $t_{e"}+2\cdot\Delta$.
\item 
Client $j$ detects that the server is corrupt, with PoM, at or before  $t_{e'}+2\cdot\Delta$ with probability $ 1-(1/N)$, where N is the total number of clients.

\item Client $j$  detects that the server is corrupt at or before $t_{e'}+2\cdot\Delta$.

\item Client $j$ receives $i$'s correct public key during epoch $e$.
\end{enumerate}
\end{theorem}
\begin{proof}
In KTACA, we allow clients to update their keys anytime during an epoch and include their updated keys in the Merkle tree in the next epoch. So, the detection process following a key update starts at the beginning of the next epoch following the update. 

First, (a) if client $i$ does not receive a valid \acrshort{str} with a \acrshort{poi} for its key within $2\cdot\delta$ time, then the client considers it an attack, and the attack is detected within $2\cdot\delta$ time after epoch $e"$ begins. Also, (b) if client $i$ does not receive an \acrshort{str} from the server in response to its \acrshort{asr} within $2\cdot\Delta$ time, then client $i$ considers it an attack and is detected within $2\cdot\Delta$ time after epoch $e"$ begins.

When a client comes online, from our design, it retrieves all missing STRs and corresponding \acrshort{poi}s for their key. In this case $i$ retrieves all STRs and \acrshort{poi}s for epochs $e^*+1$ to $e"$ and verify \acrshort{poi}s in corresponding STRs. Also, $i$ verifies the linear history of STRs. If either of these verification fails or $i$ received conflicting STRs for epoch $e"$, while retrieving directly and anonymously, $i$ detects the attack. Since $\Delta \ge \delta$, client $i$ detects the attack at or before $t_{e"} + 2\cdot\Delta$.

Second, each client receives a valid \acrshort{str} with a \acrshort{poi} for its key directly from the server and also receives an \acrshort{str} anonymously through \acrshort{asr}. If the \acrshort{str}s do not match, the server has equivocated, and the attack is detected. 


Assume $j$ receives a fake key for $i$ from the server in epoch $e$ along with a \acrshort{str} and corresponding \acrshort{poi}. Thus, to avoid detection, the server has to return the same \acrshort{str} when $j$ requests an \acrshort{str} through an \acrshort{asr}. Client $j$ sends \acrshort{asr} in $e+1$, if it is online or when it comes online ($e'$). If there are $N$ registered clients, the probability of returning the fake \acrshort{str} to the $j$ in response to an \acrshort{asr} is $1/N$. 
If the \acrshort{str}s do not match, an attack is detected, and $j$ has two conflicting \acrshort{str}s for a \acrshort{pom} at or before $t_{e'} + 2\cdot\Delta$. The victim client, who receives a fake key, detects the equivocation with a probability $1-(1/N)$.

Third, from the KTACA design, if client $j$ requests the key for client $i$, client $j$ must receive the key and \acrshort{poi} within $2\cdot\Delta$ time. However, as described in the first paragraph of the proof, the server can say that it will include the updated key in the next epoch. So, if a client does not receive \acrshort{poi} or it is invalid for epoch $e'$, the client considers it an attack at detects it at or before $t_{e'}+ \Delta$. 

If the other cases do not hold, then client $i$ receives a valid \acrshort{str} with \acrshort{poi} for its own key, and all clients receive the same \acrshort{str}s. PoI-Lemma in KT ensures that if a valid \acrshort{poi} is in the \acrshort{str} that is consistent with the \acrshort{str} $i$ has, it is consistent with the key that client $i$ published in the tree. 
\end{proof}

\section{Short-lived attack monitoring}

\label{subsec:KeyHist}
All three defenses must handle the case where \A updates the key for client $j$ after beginning an epoch $e$ and restores the correct key before epoch $e+1$. Short-lived attack monitoring detects an adversary that launches a fake key attack and quickly restores the correct key to avoid detection. This attack can also be prevented in KTCA and KTACA by allowing at most one key change per client per epoch to ensure that any new key will be in $STR_{e+1}$. Allowing at most one key change per client per epoch is reasonable because clients rarely change their keys.


We assume that the secure messaging app does not re-use keys between app re-installs, and the server does not allow clients to re-use the same key pair. Current messaging apps like Signal and WhatsApp do not re-use keys. Also, this aligns with best practices to generate a new key pair instead of re-using previous key pairs. 

When a client requests a key for a client, the server response contains the key pair, the current epoch number, and the response's signature.
To detect a rapid fake-key attack, each client maintains a \textit{key update history} for its contacts, then checks for duplicates. For example, when Alice gets a key update request for Bob, if the new key exists in Bob's history, Alice detects a rapid fake-key attack.

\label{short-lived_proof}
\begin{lemma}\label{lem:short-lived}
If \A restores the correct key for a contact after giving a fake key previously, the attack is detected instantly with \acrshort{pom} by the victim who receives the fake key.
\end{lemma}
\begin{proof}
Assume Alice has Bob's public key $K_1$.
\A sends a fake key update message (regarding Bob) to Alice with fake key $K_f$ as Bob's new public key and impersonates Bob to Alice. 
Next, Alice verifies Bob's new key update in Bob's key update history but does not find an instance of $K_f$; therefore, Alice considers it a legitimate update and adds $K_f$ to Bob's key update history. 
Later, \A terminates the attack by sending a key update message restoring Bob's correct key $K_1$ to Alice.
Once again, Alice verifies Bob's new key update in Bob's key update history and finds an instance of $K_1$ already in the file. Therefore, Alice detects the rapid fake-key attack with a \acrshort{pom} (duplicate signed keys for two different epochs.)
\end{proof}



\begin{table*}
	\centering
	\caption{Performance analysis}
	\label{tab:performance}
	\begin{tabular}{l|c|c|c|c}
		                                     &  & \multicolumn{3}{c}{Resources}                                                                    \\
		\textit{Defense}                     & Detection time           & Client side memory & \multicolumn{2}{c}{Network Traffic (per client)}                        \\
		                                     &                            &                               & (per epoch)                 & (per new connection or key update) \\
		\midrule

    	KTCA  & < 2 epochs ($1$ epoch + $2\cdot(diam(G) $+$1)\cdot\delta$)                            & $104 B$                       & $7.136 KB$                                & $1.056 KB$          \\

		AKM      & <$m$ epochs ($\approx10^\ddagger$ epochs)                            & $0$                           & $~32 KB$                                 & $~320 KB$             \\

		KTACA                        &  <1 epoch +$\Delta$                      & $104 B$                     & $33.96 KB$                                 & $1.056 KB$            \\
		

	\end{tabular}

	\vspace{3mm}

	$\ddagger$~ = detects with $0.999$ probability\\

\end{table*}

\section{Proof-of-Concept and Performance}
\label{sec:proof_of_concept_performance}

We built a Signal infrastructure in the lab to demonstrate the feasibility of all of the fake key attacks from Fig.~\ref{fig:ThreatModel}. We also built proof-of-concept prototypes to demonstrate the effectiveness of the KTCA, AKM and KTACA defenses. Each of these defenses were able to detect all of the attacks. Details are described in Appendix~\ref{sec_attack_desc}.


We evaluate our design implementations with the following parameters based on the assumptions made by CONIKS \cite{melara2015coniks}:
\begin{itemize}
\item An IM application supporting $N = 2\textsuperscript{32}$ users. 
\item Epochs occur roughly once per hour.
\item Up to 1\% of the users change or add keys per day, \ie $n  \approx 2\textsuperscript{21}$ key updates per epoch.
\item A 128-bit cryptographic security level (SHA-256, 512 bit EC-Schnorr signatures).
\item Clients have an average of 100 connections.
\end{itemize}

Table~\ref{tab:performance} reports data related to the performance of each defense, including (1) the verification delay for defenses that detect attacks, (2) the client-side memory requirements, and (3) an estimate of the total network traffic generated per client for each defense.



\subsection{Client-side Memory Requirements}

KTCA and KTACA store the previous STR at the client, requiring 104 bytes (64 for the signature, 32 for the root, and 8 for a timestamp).
AKM stores nothing.



\subsection{Client-side Network Traffic}

\paragraph{Monitoring cost}

In KTCA, a client monitors its own key binding in the tree every epoch. It requires downloading an STR and \acrshort{poi} of their key. The STR contains the root of the tree and the signature (64 bytes). The \acrshort{poi} is of size $\log_2 (N)+1$ (\ie the depth of the Merkel tree). However as all the hashes in the \acrshort{poi} do not change every epoch and if $n$ key update happens every epoch ($n<N$), then the expected number of changes in the hashes in a given \acrshort{poi} path is $\log_2(n)$. 
So, a client downloads in total $64+\log_2(n)\cdot32 = 736$ bytes.

Clients exchange \acrshort{str}s with all their contacts every epoch. A client does not send an \acrshort{str} to a contact if the client received the same \acrshort{str} from that contact. Therefore a client in total sends and receives \acrshort{str}s from 100 contacts, which totals to $100\cdot64 = 6.4$ KB. Thus, the total network data used for monitoring is $6.4 + 0.736 = 7.136$ KB per epoch.

In AKM, a client retrieves keys from the server using the Tor circuit. There is no straightforward approach to find data usage for the Tor circuit theoretically. So, we implemented an Android app that 1) creates a Tor circuit and 2) retrieves a key bundle from our implementation of the Signal server. We used the packet capture Android app and analyzed a packet trace in Wireshark. One key retrieval using the Tor circuit uses $\approx32$ KB of data (15 KB sent + 17 KB received). For monitoring in AKM, a client retrieves its key every epoch anonymously, requiring $\approx 32$ KB of data.

In KTACA, clients monitor their key, which requires downloading $64+\log_2(n)\cdot32 = 736$ bytes every epoch. Furthermore, the client audits the STR by retrieving the STR anonymously, which requires $\approx32$ KB. The total size of the network data per epoch is 33.952 KB (1.216 KB + 736 bytes + 32 KB).


\paragraph{New contact lookup and verification cost}

In KTCA, when a client communicates with a new contact, it retrieves the \acrshort{poi} for that contact, which contains $\log_2(N)+1$ hashes. A new key lookup requires downloading $32 \cdot (\log_2(N)+1) = 1.056$ KB.


In AKM, a client does not download/upload anything for a new key lookup. However, the client needs to monitor a new key lookup for some epochs. In AKM detection mode, after a client starts a new communication, it monitors a new connection's key until it has confidence that the key she received is correct. Assuming that a client wants to detect an attack with a probability of 0.999, the contact's key monitoring is needed for ten epochs.
Thus, monitoring every new contact requires $\approx32$ KB for an epoch. If the client monitors a new contact for ten epochs, the total size of the network data for new contact monitoring is $320$ KB $(32~ KB \cdot 10)$. 

In KTACA, similar to KTCA, a client downloads a \acrshort{poi} for a new contact containing $\log_2(N)+1$ hashes. A new key lookup requires downloading $32 \cdot (\log_2(N)+1) = 1.056$ KB.

\paragraph{Performance cost example}
Assume an epoch's duration is one day, a client communicates with five new contacts every month, and one of their existing contacts updates their key each month. If the client uses KTCA, the network overhead is $220.416$ KB per month ($(30\cdot7.136) + (5\cdot1.056)+1.056$ KB), and the client side memory is $104$ B. If the client uses AKM, the network overhead is $2.88$ MB per month ($(30\cdot32) + (5\cdot320)+320$ KB). If the client uses KTACA, the network overhead is $1.025$ MB per month ($(30\cdot33.96) + (5\cdot1.056)+1.056$ KB). As mentioned in Table~\ref{tab:performance},  KTCA, AKM and KTACA detect the attacks within 2 , 10 and 1 day respectively, where $\Delta<<< 1$ day.

\section{Comparison and Discussion}
\label{subsec:compare}

Table~\ref{tab:comparison} compares how the three defenses detect the MITM and impersonation attacks in Fig.~\ref{fig:ThreatModel}.
For all pair-targeted attacks, (a) and (c), \A must equivocate for the attack to succeed. KTCA and KTACA detect the equivocation and generate a \acrfull{pom}.
For client-targeted attacks, (b) and (d), if \A equivocates, KTCA and KTACA detect the equivocation and generate a \acrfull{pom}.
If \A does not equivocate, then KTCA and KTACA detect the specific key that is fake but provide no \acrshort{pom}. 
AKM detects the specific key that is fake for all of the attacks but provides no \acrshort{pom}. 

For the client-targeted impersonation attack in AKM and KTCA/KTACA (non-equivocation), only the targeted client detects the attack. 
The client must manually alert its contacts of the attack.
Future work could extend the design to notify the other victim clients through out-of-band channels automatically. 

In KTCA, all clients receive a \acrshort{pom} when \A equivocates, but in KTACA, only the victim clients receive a \acrshort{pom}. Future work could explore forwarding the \acrshort{pom} over out-of-band channels to notify additional clients of the attack.

Since some attacks are detectable without any \acrshort{pom}, a dishonest user could falsely accuse a server of a fake key attack.
Likewise, an adversary can respond to an accusation by accusing the client of making a false report.
There is no way for a third party to determine which claim is correct.

{\bf Key monitoring:} (1) Clients monitor their keys during each epoch. In KTCA and KTACA, clients retrieve an \acrshort{str} and \acrshort{poi} from the server. In AKM, clients retrieve their key via the anonymous network. 
(2) Clients also monitor their contacts' keys for both new connections and key updates. In KTCA and KTACA, clients receive an \acrshort{str} and \acrshort{poi} from the server. In AKM, clients retrieve their contacts key for $m$ epochs via the anonymous network.
(3) In KTCA and KTACA, clients audit the server for equivocation during each epoch. In KTCA, the clients compare STRs with all their contacts to confirm that everyone receives the same \acrshort{str}. In KTACA, clients retrieve an \acrshort{str} anonymously and verify that the server maintains a linear history of \acrshort{str}s.

{\bf Server network load:} The defenses differ in the new demands placed on the server. Assuming $N$ total clients, the server load is as follows. For KTCA, the server handles $N$ self-monitoring requests each epoch. 
The server maintains a key transparency log and includes an \acrshort{str} and \acrshort{poi} in the response for each key lookup. 
For AKM, the server handles $N$ anonymous self-monitoring requests for each epoch. After each key update for a client's contact, the server handles monitoring requests for that contact's key for $m$ epochs. 
For KTACA, the server handles $N$ self-monitoring requests per epoch, similar to KTCA and $N$ anonymous \acrshort{str} monitoring requests per epoch.
The server maintains a key transparency log and includes an \acrshort{str} and \acrshort{poi} in the response for each key lookup.




{\bf Deployment:} Most of the defenses require changes to existing servers. For KTCA, the server needs to support a key transparency log and distribute \acrshort{str}s and \acrshort{poi}s to clients that make a key monitoring request for their key or a contact’s key. 
For AKM, the server needs to support anonymous key requests, or the clients need to use random identifiers. In addition, clients need to use an anonymous communication network when making requests.
With KTACA the server needs to create and maintain a key transparency log and support anonymous STR requests (\eg with Tor). Anonymous STR retrieval is only done once per epoch per client to retrieve 800 bytes (736 bytes for key monitoring + 64 bytes STR). KTACA can spread out the anonymous requests over the epoch duration to reduce the load on TOR if there are many clients in the system.


Deployment can be incremental. 
For KTCA, the adopting clients need to form a connected graph to receive protection. 
For AKM and KTACA, any two clients in contact with each other can opt-in to the defense to detect fake key attacks on their connection. 

\textbf{Group chat:} For group chat, messaging applications use two different methods: (1) Treat each group message as a direct message to the receivers (Signal app), or (2) When sending a message to a group for the first time, generate a {\em Sender Key} and distribute it to each group member's device using the pairwise encrypted sessions (WhatsApp).
In both cases, the group chat will be secure if a client can verify the identity public key of the groups' contacts and have a secure pairwise connection with each group member.

\textbf{Private information retrieval}
In AKM, clients create a distinct, single-key AKR for each key they monitor to prevent deanonymization through analyzing bulk requests (see final paragraph, Section~\ref{sec:AKM_design}).
An alternate design to prevent deanonymization while allowing bulk requests (multiple key requests in one AKR) is to use private information retrieval (PIR)~\cite{chor1995private}. PIR allows a user to retrieve an item from a server in possession of a database without revealing which item is retrieved. Creating an AKR using PIR allows clients to send bulk requests in one AKR without compromising anonymity. PIR hides which key(s) are requested, and bulk queries are not visible for the server to deanonymize the requester. We did not include PIR in our design because of its high overhead. However, PIR can be used in the future if it becomes more efficient than creating a new AKR for each key request.


\textbf{Managing the signing key across multiple servers for KTCA and KTACA}
For scalability, popular messaging applications use load balancing to distribute high-volume traffic across multiple servers. For KTCA and KTACA, the servers must maintain a consistent copy of the STR across all the servers to prevent a false positive for an equivocation attack. Since the system generates an STR only once per epoch, one approach is to dedicate one secure server to store the private key, sign the STR every epoch, and distribute the signed STR to all the servers.
DKIM requires mail servers to sign individual messages, which is much more effort than signing an STR once per epoch, so we feel this is feasible for current systems. Even though clients will detect a fake key attack if the private key is compromised, the potential harm to the provider's reputation motivates the provider to secure their private key.

In KTCA and KTACA, the server has to ensure authentication between the client and server to prevent MitM from flooding the clients with fake STRs to harm the server's reputation. Currently, messaging applications use certificate pinning to hardcode a list of keys for authorized servers.

\textbf{Notifications:} 
Even with automated detection,
a significant usability issue arises regarding what to do once an attack is detected. Some potential future work directions include: 

\begin{enumerate}[topsep=1pt,noitemsep,leftmargin=1.5em]
    \item Some defenses detect a fake key attack with a proof of misbehavior but do not indicate the specific victims of the attack. What could be done? Perhaps all/some users should be notified, but it is an open problem about how to notify effectively. Another option is not to notify individual users but relay the proof to security experts or privacy advocates. If an attack were ever detected, it would make front-page headlines in the press that a major service provider had equivocated. It could result in a loss of reputation. Users could rely on expert advice given outside the app on how to respond. 

    \item Some defenses detect an attack on specific users without generating a proof to convince others. How to notify the victim and what to recommend is another open problem. One might mimic what the current apps do when the manual key ceremony fails and compare this to other alternatives. Also, users could be encouraged not to share sensitive information with a contact. 

    \item Organizations using secure messaging (\eg political campaigns, news organizations) might prefer to alert security admins instead of users when an attack is detected.

    \item Provide UI indicators that a connection is secure after monitoring a contact for sufficient time instead of educating users about fake key attacks.

\end{enumerate}


\textbf{Recommendation:} We believe KTACA provides the best combination of features of all three defenses. It has the strongest security properties, provides a PoM where possible, and has a low detection delay. One hurdle to deployment is that KTACA requires an anonymous network, and some countries block Tor, but these countries are likely to also block secure messaging apps. KTCA does not rely on an anonymous network, but is vulnerable to a graph partitioning attack. Where deployment costs are a concern, AKM may be preferred because it only requires changing clients to support anonymous requests. However, AKM does not provide PoM and imposes significantly higher detection delay.




\section{Conclusion}


We designed three automated key verification defenses to detect fake key attacks in secure messaging applications. The defenses enable fake key attacks to be automatically detected, which relieves users from manually comparing key fingerprints to detect attacks. Since prior studies show that most users do not manually verify connections, these defenses can fill this void. In addition, the defenses may deter \A from launching fake key attacks.
However, vulnerable users can still perform manual key verification if they do not trust the automated system or want increased assurance.

\begin{acks}
We thank the reviewers and our shepherd, Lucjan Hanzlik, for their helpful feedback on the final version of the paper.
This work was partially 
supported by the \grantsponsor{123}{National Science Foundation}{https://www.nsf.gov} under Grant No. \grantnum{123}{CNS-1816929} and by the \grantsponsor{124}{Comcast Corporation}{https://corporate.comcast.com}. The opinions expressed in the paper are those of the researchers and not of their universities or funding sources.

\end{acks}

\bibliographystyle{ACM-Reference-Format}
{
\bibliography{signal}
}
\appendix
\section{Appendix}

\subsection{Implementation prototypes}

Our Signal infrastructure consisted of a server and five Signal clients (installed on Android mobile phones) for performing attacks and experimenting with defenses. The machine running the Signal server had Intel (R) Core (TM) i7-7700K CPUs @ 4.20 GHz, with 16 GB RAM. We used the open-source Signal server \texttt{OWS 1.88} and clients using \texttt{Signal-Android 4.23.4}.

\label{section:protoypes}
{\bf Fake Key Attacks }
\label{attacks_implement}
We modified the Signal server to perform fake key attacks (both MITM and impersonation). 
We added a client module to the server using the Signal client CLI\footnote{https://github.com/AsamK/signal-cli}. 
It performs the standard ratcheting encryption done in normal Signal clients. 
We modified the server to hand out fake keys to clients. 
When the clients send encrypted messages to each other through the server, the server redirects the messages of the victim(s) to the client CLI module where the message is decrypted for the MKS to access it, and then re-encrypted and sent to the victim in the case of a MITM attack.
Our experiment demonstrates that all fake key attacks described in Section~\ref{sec_attack_desc} are possible in the current Signal server implementation.

{\bf KTCA prototype:} We used the transparency log prototype from CONIKS~\cite{coniksjava} and built the STR verification protocol on top of it. In addition, we modified the Android Signal client to add support for \textit{STR Verification Messages (SVM)}, which clients use to exchange STRs with their contacts to verify their consistency during every epoch. SVMs are routed through the Signal server in the same way that current Signal messages are communicated, but with an additional flag set in the message body. The receiving clients check this flag on every incoming message before displaying it in the chat window. We ran a simulation to measure the overhead of checking this flag on one million messages. The cost is 0.622 milliseconds per message before displaying the message in the chat window,  which has a negligible impact on the user experience.


{\bf AKM prototype} We modified the Android Signal clients to support Anonymous Key Retrieval using the Tor Android Library.\footnote{https://github.com/jehy/Tor-Onion-Proxy-Library}
We use \texttt{TinyWebServer} on the client as a local server where the identity key of the user is available, then proxy the server's traffic via a Tor SOCKS proxy to create a hidden service. 
We used one computer as a Signal server for launching fake key attacks.
For AKR, our Signal clients use Tor to anonymously query identity keys. Our requests did not include the sender's unique identity, and Signal clients sent such key retrieval requests to our Signal server.

\begin{figure}
\centerline{\includegraphics[scale=0.3]{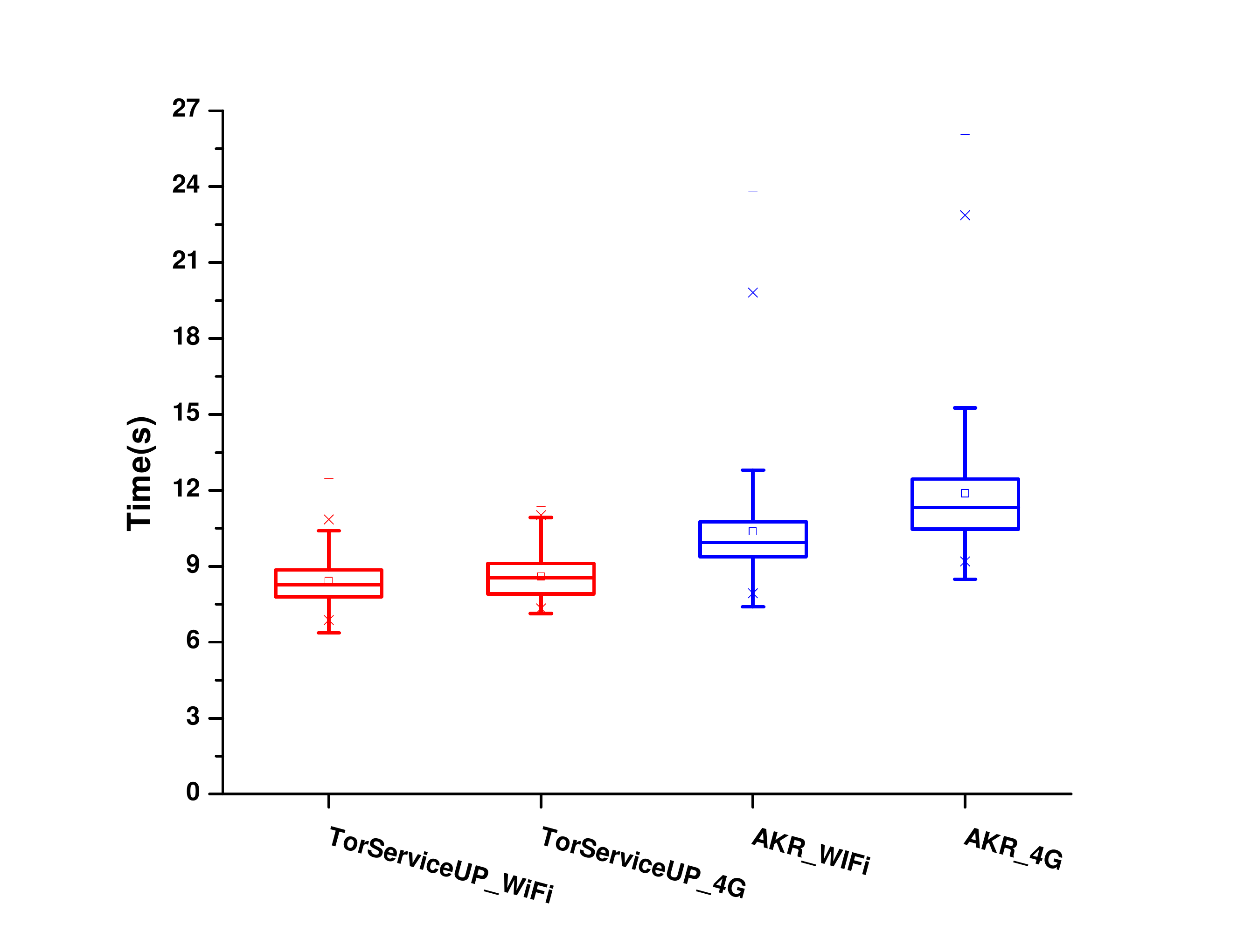}}
\caption{Time taken to a) retrieve keys through Tor circuit and b) launch the Tor service on a client. }
\label{fig:Tor_Service_Up}
\end{figure}

We also conducted an experiment measuring the time cost of accessing keys via Tor. We requested a key bundle through Tor $100$ times per day for ten days using both WiFi and 4G networks. The time includes setting up a Tor circuit and retrieving keys from the server using that circuit. 
Fig.~\ref{fig:Tor_Service_Up} shows the results of our experiment. 
Over WiFi, it took an average of $10.4$ seconds to retrieve keys from the server using Tor and an average of $11.88$ seconds over a 4G network.

{\bf KTACA prototype} We leveraged the key transparency logs and anonymity implementations from KTCA and AKM respectively. The Signal server maintains a transparency log, and the modified Android clients retrieve STRs anonymously from the server during every epoch.

{\bf Prevention mode prototype} We need an out-of-band channel. To accomplish this, a variety of third-party services exist that have differing deployability and security properties. Example third party services include SMS text messages, email, and Tor hidden service \emph{etc}. As we already use TOR in our defense and TOR services are more secure compared to other channels, we use them for our prototype. It is a TCP-based network service that is accessible only via the Tor network. We explain in detail how we use it for out-of-band verification of keys in Appendix \ref{app:Tor}.

\subsection{Prevention mode}

The three defenses discussed in the paper detect fake key attacks. However, in some scenarios, prevention is desirable (\ie clients detect the attack before they communicate). This subsection presents a design that extends any of the defenses to prevent fake key attacks.

\subsubsection{Assumptions}
We say that two clients have a $\delta$-delay channel if communication between them on this channel is reliable with maximal delay $\delta$. We assume there exists an out-of-band $\delta$-delay channel between clients and their contacts. Note that this out-of-band channel is only used for prevention mode and is authenticated by the protocol as explained in the design subsection.

We can automatically setup these out-of-band channels between two users and authenticate these using the technique described below. Various third-party services with differing deployability and reliability properties can be used to confirm a key update over an out-of-band channel. Some candidates for out-of-band connections are TCP/IP messages, SMS, email, or a Tor hidden service. A variety of systems (Signal, WhatsApp, Skype) already support direct (peer-to-peer) video calling, with features to handle Network Address Translation (NAT) traversal. Using smartphones presents a challenge with IPv4 because a phone's IP address can change due to mobility. With IPv6 in cellular networks, a phone's IP address is unique and constant. According to the Internet Society's 2018 State of IPv6 Deployment, 80\% of smartphones in the US on the major cellular network operators use IPv6~\cite{ipv6InternetSociety}.
   
\subsubsection{Design}

Prevention mode performs the basic defense process along with the following additional strategy.
\begin{enumerate}
    \item When a client connects with a new client, it waits until the key verification, described in each defense, completes before sending a message. The waiting period before a client can start communication varies in each defense based on their detection times. 
    \item After successfully verifying a new connection, clients use the existing secure channel to establish a shared key for sending authenticated messages to the new connection over an out-of-band channel. \textit{This channel will be used only for verifying key updates.} 
    \item When clients receive a key update for their contact, they use the authenticated out-of-band channel to confirm the key update from their contact directly. If a clients do not receive a response from their contact within $2\cdot\delta$ time, they consider it an attack.
    This out-of-band channel is authenticated by a MAC using keys established in Step 2. Therefore, it defends against data modification by on-path attackers for existing connections. 
\end{enumerate}

We note that the out-of-band connections are only needed to reduce the delay of the prevention mechanism. A client could use step 1) for both new connections and key updates, with the penalty being significant delay (one epoch for KTCA and KTACA, and $n$ epochs for AKM). Using steps 2) and 3) enables prevention with a much shorter delay of $2\cdot\delta$.

We also note that prevention mode can be turned on by individual users or for individual conversations. For example, a user who has a particular need could turn on "sensitive" mode, which triggers a UI that informs them that the conversation can have further protection if they call their contact or communicate out of band, asking them to likewise turn on sensitive mode. Once this is done, the clients can perform prevention mode through out-of-band connections and show this feature is enabled in the UI.

\subsubsection{Prototype} 

For an out-of-band channel in prevention mode, we used Tor services. If there is a key update, it retrieves the key from the corresponding client's Tor service. We simulated this key retrieval through TOR service, where a user connects to the Tor circuit and retrieves keys, 100 times per day for ten days. On an average it requires $\approx40$ KB of data. We also present the time it takes to launch the TOR service and retrieve a key through Tor in Fig.~\ref{fig:Tor_Service_Up}. 

The prevention mode in all defenses additionally stores the identity and symmetric key pair of the out-of-band channel for every contact. Assuming we use \textit{P2P} on IPv6 or a Tor service address, identity size is 16 bytes, and the symmetric key size is 32 bytes, which adds up to 48 bytes. From our assumption of 100 friends, the total client-side memory usage is $4.8$ KB.

\subsubsection{Analysis}

\textit{Attacks on a new connection} are prevented because clients do not communicate until they verify the keys using basic mode, which takes up to
\begin{itemize}
    \item $t_{e'} + 2\cdot(diam(G)+1)\cdot\delta$ time in KTCA, where $t_e$ be the time of the beginning of an epoch $e$, assuming a client connects with a new client during epoch $e$ and it comes online again in epoch $e'$ where $e<e'$. 
    \item $m$ epochs in AKM. Assuming an attack on only one contact, it takes an $\approx 10$ epochs to prevent the attack with $0.999$ probability (using $1-(1/(c+1))^m$ from Theorem~\ref{thm:AKM}).
    \item $2\cdot\Delta$ time in KTACA after the beginning of the next epoch to verify a new contact's key.
\end{itemize}

\textit{Attacks on an existing connection} are  prevented by immediately verifying a key update using the authenticated out-of-band connection, which takes up to $2\cdot\delta$ delay as described in the below Theorem.

\begin{theorem}\label{thm:prevention}
In prevention mode, a fake key update is detected, within $2\cdot\delta$ delay, before any further communication occurs.
\end{theorem}
\begin{proof}
When Alice receives a key update message from Bob, she contacts him through the out-of-band channel to confirm the correctness of the key update. Bob verifies the key update to help Alice detect the attack. As the communication over the out-of-band channel is authenticated using a MAC, \A can only read or drop the communication. Therefore, if Alice does not receive a key update confirmation within $2\cdot\delta$ time, she considers it an attack.  
\end{proof}
\subsection{Special-purpose Monitoring}
\label{sec:heuristics}
This section describes two special-purpose monitoring checks that augment the three defenses presented earlier by helping them achieve their connectivity assumptions. These monitoring checks detect scenarios where assumptions are violated that the defenses were not designed to address.
The two scenarios covered by these checks are 
(1) an adversary distributes fake keys to a high percentage of a client's contacts, and
(2) an adversary blocks secure communication between a client and all of its contacts.

\subsubsection{Mass key update monitoring}
\label{subsec:MassKeyUpdate}

Mass key update monitoring detects when a significant number of contacts for a client update their key over a short duration of time. The purpose of this defense is to detect a client-targeted MITM attack on Alice's existing connections or the impersonation of all of Alice's existing contacts to Alice. To succeed, the adversary needs to send fake key updates to Alice for all her contacts. However, it is improbable that all of Alice's contacts would simultaneously re-install the application and generate such a burst of key updates.

If a {\it naive adversary} sends all the fake key updates at once, Alice can trivially detect the attack. A more {\it stealthy adversary} could spread out the key updates over time, possibly using a just-in-time approach to initiate a key update as each contact attempts to exchange a new message with Alice. Alice detects the attack 
by observing that more than a threshold number of her contacts change their keys during the last $t$ seconds.

\subsubsection{Isolation monitoring}
\label{subsec:Isolation}

Isolation monitoring detects client isolation from all of its contacts.
Suppose an adversary launches a client-targeted impersonation attack on Alice by impersonating her to her existing contacts. In that case, the adversary must isolate Alice by blocking all her outgoing messages and not forwarding incoming messages to her for these existing contacts.
If Alice cannot connect to any of her existing contacts during an epoch, a client-targeted impersonation attack may be in progress. It is also possible that she has lost overall Internet connectivity, so any detection approach must verify that Alice is isolated but still has Internet connectivity.

Alice regularly monitors for complete isolation to detect this attack by sending a connection verification message to each of her contacts during each epoch. 
If none of her contacts respond within an epoch, Alice is alerted. 
Note that Alice should only consider herself non-isolated if she receives a response from a contact whose key has not changed during the epoch. This precaution is necessary since an adversary can issue a fake key update message and then impersonate one of her contacts to issue a fake connection verification message. 

To make the isolation monitoring process efficient, the client divides an epoch into $x$ subintervals, where $x$ is the client's number of contacts. During each sub-interval, the client selects a random contact and sends them an isolation monitoring query. Upon receiving a response from a query sent during a previous sub-interval, the client considers itself as non-isolated, stops monitoring, and resumes monitoring at the next epoch. In case no response is received throughout all sub-intervals of an epoch, the client assumes it is isolated from its contacts by the server.

\subsection{Key Transparency Proofs}
\label{sec:kt_proofs_appendix}

\begin{theorem} \label{thm:KTCA1}
Let $i,j$ be benign clients which are nodes in the connected graph $G_e$ (of $\delta$ connected benign clients during epoch $e$), and let $t_e$ denote the time that epoch $e$ begins. Assume that at time $t'$ within the epoch $e$, client $j$ asks the server for the key of $i$. Then one of the following holds:
\begin{enumerate}
\item Client $i$ detects that the server is corrupt, at or before $t_{e}+2\cdot\delta$. 
\item All clients in $G_{e}$ detect that the server is corrupt, with PoM, at or before $t_{e}+2\cdot(diam(G_e)+1)\cdot \delta$. 
\item Client $j$ either detects that the server is corrupt at or before $t'+2\cdot\delta$ or receives the correct public key used by $i$ during epoch $e$. 
\end{enumerate}
\end{theorem}

\begin{proof}
First, if client $i$ in $G_e$ does not receive an \acrshort{str} with \acrshort{poi} for its own key within $2\cdot\delta$ time or receives an invalid \acrshort{poi} or \acrshort{str}, then the client considers it an attack. Client $i$ detects this corrupt behavior at or before $t_{e} + 2\cdot\delta$.

The second case describes the detection of equivocation.  If clients $i,j$ receive different \acrshort{str}s, then one or more clients in the graph receive $STR_1$, and one or more clients in the graph receive $STR_2$. As the graph $G_e$ consists of $\delta$-connections between clients, there must exist at least one edge connecting two clients where one client receives $STR_1$, and the other client receives $STR_2$. Once these two clients exchange the \acrshort{str}s, a \acrshort{pom} exists. Then the \acrshort{pom} is forwarded to all clients in $G_e$. 

The maximum time that is incurred for detection by all clients is the case when two farthest clients (the ones which are $diam(G_e)$ apart) receive the conflicting \acrshort{str}s and all the intermediate clients have not received \acrshort{str}s from the \acrshort{ms}.
The detection time is computed as follows. At the beginning of an epoch, within $2\cdot\delta$ time, client $i$ receives $STR_1$, client $j$ ($diam(G_e)$ apart from $i$) receives $STR_2$ and all the intermediate clients do not receive \acrshort{str}s. Clients $i$ and $j$ send their \acrshort{str}s to their neighbors and in turn they also forward the \acrshort{str}s to their neighbors. Within $diam(G_e)\cdot\delta$ time the node which is $diam(G_e)/2$ apart from both $i$ and $j$ receives both the \acrshort{str}s and generates the PoM. Eventually, in $diam(G_e)\cdot\delta$ time this PoM is propagated within the whole $G_e$ network and all the clients in $G_e$ become aware of the equivocation including $i$ and $j$. Thus, to detect equivocation the maximum time is $2\cdot\delta + diam(G)\cdot \delta +  diam(G_e)\cdot \delta$, which is $2\cdot(diam(G_e)+1)\cdot\delta$.
Third, from the KT design, if any client $j$ asks for the client $i$ key, client $j$ must receive the key of $i$ and \acrshort{poi} within $2\cdot\delta$ time. If a client does not receive \acrshort{poi} or it is invalid, the client considers it an attack and detects it at or before $t'+ 2\cdot\delta$.

Otherwise, client $j$ receives a valid \acrshort{poi} and \acrshort{str}. If Case 1 and 2 do not hold, then client $i$ receives a valid \acrshort{str} with \acrshort{poi} for its own key, and all clients receive the same \acrshort{str}. PoI-Lemma in KT ensures that if a valid \acrshort{poi} is in the \acrshort{str} that is consistent with the \acrshort{str} $i$ has, it is consistent with the key that client $i$ published in the tree. \end{proof}

\begin{theorem}\label{thm:KTCA2}
Let $j$ be a benign client that is a node in the connected graph $G_e$, and $i$ be a benign client in the connected
graph $G_{e"}$ and $G_{e^*}$ where $e^*<e<e"$, and let $t_e$ denote the time that epoch $e$ begins. Assume that client $j$ requests $i$'s key, which is included in $STR_e$, from the server at time $t'$ within epoch $e$. Then one of the following holds:
    \begin{enumerate}
    \item Client $i$ detects that the server is corrupt, at or before $t_{e"}+2\cdot\delta$.
    \item All clients in $G_{e"}$ detect that the server is corrupt, with \acrshort{pom}, at or before $t_{e"}+2\cdot(diam(G_{e"}+1))\cdot \delta$.
    \item Client $j$ either detects that the server is corrupt at or before $t'+2\cdot\delta$ or 
    receives the correct public key used by $i$ during epoch $e$.
    \end{enumerate}
\end{theorem}
\begin{proof}
This proof builds upon the arguments in the previous theorem's proof. 

First, if client $i$ in $G_{e"}$ does not receive \acrshort{str}s with \acrshort{poi}s for epochs $e^*+1$ to $e"$ for its own key within $2\cdot\delta$ time or receives any invalid \acrshort{poi} for a corresponding \acrshort{str}, then client $i$ considers it an attack. Client $i$ detects this corrupt behavior at or before $t_{e"} + 2\cdot\delta$

The second case describes the detection of equivocation. From our assumption, each epoch has more than $N/2$ clients. It ensures that there is an overlap of at least one client in $G_e$ and $G_{e"}$. To equivocate, the server has to give a different \acrshort{str} in epoch $e"$ to the overlapping clients than it gives to client $i$. As described in Theorem~\ref{thm:KTCA1}'s proof, if clients receive different \acrshort{str}s in $G_{e"}$, it is detected with \acrshort{pom} by all clients in $G_{e"}$ at or before $t_{e"}+2\cdot(diam(G_{e"}+1))\cdot \delta$.

Third, from the KT design, if any client $j$ asks for client $i$'s key, client $j$ must receive the key of $i$ and \acrshort{poi} within $2\cdot\delta$ time. If a client $j$ does not receive \acrshort{poi} or it is invalid, the client considers it an attack and detects it at or before $t'+ 2\cdot\delta$.

Otherwise, client $j$ receives a valid \acrshort{poi} and \acrshort{str}. If Case 1 and 2 do not hold, then client $i$ receives a valid \acrshort{str} with \acrshort{poi} for its own key, and all clients receive the same \acrshort{str} in epoch $e"$. PoI-Lemma in KT ensures that if a valid \acrshort{poi} is in the \acrshort{str} that is consistent with the \acrshort{str} $i$ has, it is consistent with the key that client $i$ published in the tree.
\end{proof}

\subsection{AKM proofs}
\begin{theorem} 
\label{thm:AKM2}
Let $t_e$ denote the time that epoch $e$ begins. Assume $f$ clients receive a fake key and $r$ clients receive a real key for client $i$, during epoch $e$. Then following holds: 
\begin{enumerate}
    \item client $i$ detects that the server is corrupt at or before $t_{e+m}$ with probability $1-((r+1)/(r+f+1))^m$, where $m$ is the number of epochs that the contacts monitor the client's key following a key update. If clients disconnect during $m$ epochs, then the probability of detection is $1-\prod\limits_{i=1}^{m} max((real_i+1)/(real_i+fake_i+1), owner_i)$, where $real_i$ and $fake_i$ are the number of contacts online who received real and fake key in epoch e ($real_i<r$, $fake_i<f$), and $owner_i$ is a $1$ or $0$ depending on whether the owner is offline or online.
    \item at least one client detects that the server is corrupt at or before $t_{e+m}$ with probability $1-(1/{(n+1) \choose f})^m$, where $m$ is the number of epochs that the contacts monitor the client's key following a key update. If clients disconnect during $m$ epochs, then the probability of detection is $1-\prod\limits_{i=1}^{m} 1/{n' \choose fake_i}$, where $fake_i$ are the number of contacts online who received a fake key in epoch e ($fake_i<f$), $n'$ is the number of total online clients monitoring $i$'s key including $i$, and $owner_i$ is a $1$ or $0$ depending on whether the owner is offline or online.
\end{enumerate}

\end{theorem}
\begin{proof} 

When \A presents a fake key for Alice to $f$ of its contacts in an epoch and continues the attack for at least $m$ epochs, 1) the $f+r$ contacts monitor Alice's key for $m$ epochs, and 2) Alice monitors its key. All of these requests are indistinguishable from each other. To avoid fake key detection, \A has to deliver the correct key to Alice and $r$ contacts in every epoch and the fake key to its $f$ contacts.

\A knows it will receive $r+f+1$ requests, and only $r+1$ of them should return the real key, and $f$ should return the fake key.
In the first case, client $i$ detects the attack when the server gives it a fake key. The probability that server can hand out a real key to client $i$ in an epoch is $(r+1)/(r+f+1)$ and the probability of providing real key to Alice for $m$ epochs is $((r+1)/(r+f+1))^m$ 

In the second case, \A has ${(r+f+1) \choose f}$ choices for plausible ways to distribute the keys, and only one of them is correct. During each epoch, the probability of making the right choice is $(1/{(r+f+1) \choose f})$. So as the $f+r+1$ clients monitor Alice's key for $m$ epochs, the probability of \A making the right choice to avoid detection is $(1/{(r+f+1) \choose f}))^m$. When subtracted from 1, this is the probability where at least one client detects the attack within $m$ epochs, validating our second case.

If Alice disconnects for any of those $m$ monitoring epochs, \A can distribute fake keys with no detection by $i$ during those epochs. Suppose at every epoch $i$ during the monitoring interval we have $real_i$ and $fake_i$ contacts online where ($real_i<r$, $fake_i<f$), and $owner_i$ is a $1$ or $0$ depending on whether the owner is offline or online. If the owner is offline and \A knows this, \A can hand out fake keys without worrying about detection from $i$ (Case 1). The probability of avoiding detection from $i$ is the product of the probability during each epoch. So the detection probability is:  $1-\prod\limits_{i=1}^{m} max((real_i+1)/(real_i+fake_i+1), owner_i)$. To avoid detection from any client, which is second case, the server has to handout real and fake keys to corresponding online clients during every epoch, and its probability is $\prod\limits_{i=1}^{m} 1/{n' \choose fake_i}$

If the attack is short-lived (less than $m$ epochs), and \A restores the original key, the attack is detected using the \textit{short-lived attack monitoring} basic defense as soon as \A restores the original key (described in~\S\ref{short-lived_proof}).
\end{proof}

\subsection{Tor Hidden Service for out-of-band key verification}
\label{app:Tor}

Our modified android clients host their own Tor service
when they install the application. When Alice and Bob \textit{first} establish a secure connection with each other, they automatically share the URL of their TOR services with each other. Later, when Alice receives a key update message from Bob, she contacts his Tor service to confirm whether the key update is legitimate. Bob responds with his correct key via the Tor service to help Alice prevent the attack.

Note that clients run TOR services online only when they update their keys (probably by re-installing the app). But they need to use the same Onion URL and need TOR service's private key for that. We use the messaging application's backup mechanism to store the TOR service's private key that can persist across re-installs or different phones. Whenever clients backup their messages, we every time append a message at the top with the label "TOR\_service\_details". Whenever user re-installs the app, they get an option to restore chat messages while logging into a previous account. In this step, we automatically retrieve the TOR service's private key and host the client's TOR service in the background. \\
The lack of a response is an indicator to Alice that the key update for Bob is an attack because if Bob does not update the key it will not turn on its TOR service.
\end{document}